\documentclass[aap]{stylefile}
\usepackage{amsthm}
\usepackage{amsmath}
\usepackage{amsfonts}
\usepackage{amssymb}
\usepackage[numbers,sort]{natbib}
\usepackage{mathtools}
\usepackage[shortlabels]{enumitem}
\usepackage{algorithm}
\usepackage{algorithmic}
\usepackage[colorlinks,citecolor=blue,urlcolor=blue]{hyperref}
\usepackage{graphicx}
\usepackage{ragged2e}
\usepackage{bm}
\usepackage{multicol}
\usepackage{multirow}
\usepackage{pifont}
\usepackage[nameinlink]{cleveref}
\usepackage{tikz}
\usepackage{dsfont}
\usepackage{color,colortbl}
\usepackage{pdflscape}
\usepackage{pifont}
\usepackage{wrapfig}
\usepackage{caption}
\usepackage{subcaption}

\usepackage[font=small,labelfont=bf]{caption}

\definecolor{Gray}{gray}{0.85}
\definecolor{LightCyan}{rgb}{0.88,1,1}

\startlocaldefs
\theoremstyle{plain}
\newtheorem{theorem}{Theorem}
\newtheorem{lemma}{Lemma}
\newtheorem{corollary}{Corollary}[theorem]
\newtheorem{proposition}{Proposition}

\theoremstyle{definition}
\newtheorem{definition}{Definition}

\theoremstyle{remark}
\newtheorem*{remark}{Remark}

\endlocaldefs

\allowdisplaybreaks

\begin{document}

\begin{frontmatter}
\title{Matching Queues with Abandonments in Quantum Switches: Stability and Throughput Analysis}

\begin{aug}
\author[A]{\fnms{Martin} \snm{Zubeldia}\ead[label=e1]{zubeldia@umn.edu}},
\author[B]{\fnms{Prakirt R.} \snm{Jhunjhunwala}\ead[label=e2]{rj2122@columbia.edu}},
\and
\author[C]{\fnms{Siva Theja} \snm{Maguluri}\ead[label=e3]{siva.theja@gatech.edu}}
\address[A]{
University of Minnesota,
\printead{e1}}

\address[B]{
Columbia University,
\printead{e2}}

\address[C]{
Georgia Institute of Technology,
\printead{e3}}

\end{aug}

\begin{abstract}
Inspired by quantum switches, we consider a discrete-time multi-way matching system with two classes of arrivals: requests for entangled pair of qubits between two nodes, and qubits from each node that can be used to serve the requests. An important feature of this model is that qubits decohere and so abandon over time. In contrast to classical server-based queueing models, the combination of queueing, server-less multi-way matching, and (potentially correlated) abandonments make the analysis a challenging problem. The primary focus of this paper is to study a simple system consisting of two types of requests and three types of qubits operating under a Max-Weight policy.
  
In this setting, we characterize the stability region under the Max-Weight policy by adopting a two-time scale fluid limit to get a handle on the abandonments. In particular, we show that Max-Weight is throughput optimal and that it can achieve throughputs larger than the ones that can be achieved by non-idling policies when the requests are infinitely backlogged. Moreover, despite the use of the Max-Weight policy, we show that there can be a counter-intuitive behavior in the system: the longest requests queue can have a positive drift for some time even if the overall system is stable.
\end{abstract}


\end{frontmatter}

\tableofcontents

\section{Introduction}\label{sec:intro}

During the last century, telecommunication networks' theoretical analysis, design, and control have evolved together with the underlying technology. Starting with the circuit-switched telephone networks of the early twentieth century, through the boom of the packet-switched network that is the Internet, the advent of quantum computers requires a new type of switched network \cite{quantumMagazine,shortcuts,quantumBook}. In quantum computers, the qubit replaces the bit as the minimum unit of information. Qubits act as a richer unit of information that allows new computational applications and a new way of communicating by \emph{entangling} the state of qubits at different nodes. This \emph{entanglement} is a fundamental feature of quantum mechanics that has no counterpart in classical mechanics. When two qubits are entangled, altering the state of one alters the state of the other. A Quantum network enables distribution of entangled qubits between users, and several architectures propose implementing such networks \cite{networks1,networks2,networks3,networks4,networks5,networks6}.

In this paper, we focus on the most fundamental component of a quantum communication network: a single quantum switch that connects two or more nodes in a star topology. In this kind of switches, connections between nodes are established in two steps: a preparation step, and a connection step. First, to prepare for connection requests, pairs of maximally entangled qubits (also know as Bell pairs, or EPR states) are constantly being generated between the switch and each of the connected nodes. Generated pairs are stored in quantum memories, with one qubit of each pair stored at the switch and the other one at the corresponding node. Due to a phenomenon called \emph{quantum decoherence}, in which the environment interferes with the qubit and eventually alters its state, the stored entangled qubits become useless after some time, and are thus lost. Second, when a request for a connection between two nodes is made, if the switch has qubits entangled with each of the nodes, then it performs a Bell-state measurement between those two qubits, entangling them with positive probability (otherwise, the qubits are lost and the request is not served). If the entanglement swapping is successful, this also results in the entanglement of the corresponding two qubits at the nodes, that can be used for quantum communication between them using teleportation and for other applications. On the other hand, if the switch has no such two qubits for entanglement swapping stored, the request must wait for their generation.

A discrete-time model for a Quantum switch based on a network of matching queues was introduced in \cite{TowsleyStochAnalysis}. In this model, the quantum memory in each node is modeled as a queue with random arrivals (the continuous generation of entangled qubits, which is random in nature \cite{qubitGenerationModel}) and abandonments (due to decoherence). We allow for general arrival distributions in each time slot, and assume that the each qubit abandons after a geometrically distributed amount of time (as in \cite{Fittipaldi23}), mimicking the natural exponential decoherence of quantum states \cite{decoherence}. Furthermore, we assume that a random number of requests for connections between two nodes arrive to the switch in each time slot, and that these connections wait in a queue at the switch until they are completed (that is, they do not abandon). Finally, the entanglement attempt of two qubits through a Bell-state measurement in response to a request for a connection corresponds to a three-way matching among the request and the pair of qubits if the entanglement is successful, and to a two-way matching between the pair of qubits if the entanglement is not successful. In the latter case, the request remains unfulfilled. The central control problem in a Quantum switch is the following. In particular, to simplify the analysis, we assume that there is no limit in the number of Bell-state measurements that can be made in each time slot. When a new entangled qubit is generated, and there is more than one type of request waiting, the switch needs to decide which connection to create. This is a matching problem among all involved queues. 

Although the model described above is inspired by quantum switches, it also models certain Assemble-to-Order systems with perishable inventories \cite{graves1982application}. For such systems, the qubits would correspond with perishable components, and the requests for connections would correspond for requests to assemble items that require two of the perishable components. Moreover, the failure of the Bell-state measurement that entangles two qubits to serve the request (which results in the qubits being lost) would correspond to a failure of a quality control check (which results in the parts being lost). This close connection between these seemingly very different systems make our conclusions for quantum switches also applicable to such Assemble-to-Order systems with perishable inventories.

Our goal in this paper is to find the stability region of a quantum switch in a W-topology, which corresponds to a switch with three nodes, but where there are only requests for connections for two out of the three possible pairs. Here, when qubits arrive, and there is more than one three-way matching possible, we use a Max-Weight policy on the requests queue lengths to decide which matching to attempt.
        
In scheduling of classical queues, the stability region of a system operating under a Max-Weight policy is generally determined by the convex hull of the set of service schedules, that is, of the set of possible departure vectors from the different queues. If the set of schedules is randomized in each time slot (even in a Markovian way), independently of which one is actually used, the stability region can still be determined by considering the steady-state expectation of the available schedules. A quantum switch in contrast involves matching queues, where the ``service schedules'' in each time slot are determined by the available qubits, which can be stored over time to fulfill future requests until they decohere and abandon the system. Therefore, the set of possible schedules is determined by a Markov chain with dynamics that depend on which schedule is used. In particular, this means that there is no straightforward way of taking expectations and obtaining the stability region, as the steady-state distribution of the schedules now depends on the arrival rates in an unknown way. Also, the stability region cannot be obtained by ``saturating'' the queues as in the queueing networks introduced in \cite{saturation}.

\subsection{Our contribution}
The main contribution of our paper is to show that the stability region of a quantum switch in a W-topology under Max-Weight can be characterized in terms of the throughputs obtained when the requests are either completely or partially infinitely backlogged. In a completely backlogged system, both kinds of requests are infinitely backlogged and the qubit queues make all possible matchings in all slots, simplifying the throughput analysis. In the partially backlogged system, one type of requests are infinitely backlogged, while the other type are not. This backlogged system is equivalent to a lower dimensional system, viz., a Quantum Switch in a Y-topology. Therefore, we first characterize the stability region under a Y-topology and use it as a stepping stone to study the W-topology. The contributions of our paper are as follows:

\begin{enumerate}
    \item \textbf{Y-topology:} A Y-topology consists of two qubit queues and one request queue. This corresponds to a quantum switch with only two nodes.
    \begin{itemize}
        \item[(i)] Studying the capacity of the Y-topology further reduces to the case when the requests are infinitely backlogged, leading to a two-way matching system.
        We obtain the stability region and average matching rate (throughput) of a two-way matching queue with (potentially correlated) abandonments, and explore its connections to a single-server queue.
        \item[(ii)] For the Y-topology, we obtain the stability region and the system's throughput. In particular, we show system's stability whenever the arrival rate of connection requests is lower than the throughput of an appropriate two-sided queue with abandonments.
    \end{itemize}
    
    \item \textbf{W-topology:} For the system in the W-topology, in order to keep the problem tractable, we assume that both the abandonment probabilities and the three-way matching success probabilities are homogeneous. In this setting, we show the following:
    \begin{itemize}
        \item[(i)] When both requests queues are infinitely backlogged, we characterize the sum of the throughputs of both types of requests, which is the same for all non-idling matching policies. Moreover, we show that the throughput of each request queue is uniformly bounded away from zero under any non-idling matching policy.
        \item[(ii)] When only one request queue is infinitely backlogged, we obtain the capacity region for the other queue when we give strict priority to the infinitely backlogged one. Moreover, we obtain the throughput of the infinitely backlogged one, as a function of the arrival rate of the other one.\footnote{Even if we give strict priority to the infinitely backlogged request queue, its throughput still depends on the arrival rate to the other queue. This is because, when there are no qubits of type 1 and there are qubits of type 2 and 3, the latter two would still be matched with a request of type 2 (if there are any), taking away qubits of type 2 that could be used to match with requests of type 1 in the future.}
        \item[(iii)] When none of the requests queues are infinitely backlogged, we show that the Max-Weight policy is throughput optimal, and that its stability region is the convex hull of all throughputs obtained in the completely and partially infinitely backlogged cases. This implies that considering only the infinitely backlogged system doesn't give the capacity region in  quantum switches.
        \item[(iv)] We use a two time-scale fluid limit approximation to characterize the transient behaviors. In particular, we show that, in some cases, the drift of the largest requests queue can be positive for some time under the Max-Weight policy, even when the system is stable.
    \end{itemize}
   \item \textbf{Methodological contribution:} We obtain the above results using two different proof techniques, both of which may be of independent interest to study more general matching networks.  
\begin{itemize}
    \item[(i)] We obtain the results on the Y-topology using a multi-step variant of the Foster-Lyapunov theorem. To overcome the challenge of the ``service schedules'' being dependent on the requests queues, we construct two coupled processes. In one of the processes matchings are always attempted regardless of requests, and in the other, no matchings are attempted (pure abandonment). We show that these processes stochastically dominate the original one. We then use the multi-step Lyapunov argument to show stability of the coupled processes, implying the stability of the original process. 
    \item[(ii)] We obtain the results on the W-topology using fluid limit arguments. The key challenge is in establishing the fluid limit for two reasons. 
    The first one is that we have a discontinuous drift, and thus we cannot apply the well-known fluid limit theorems that rely upon the Lipschitz continuity of the drift. The second one is that we have an averaging effect in the drift, as the qubit queues evolve in a faster time scale than the requests queues.
    We overcome these two challenges following roughly the same approach as in \cite{positiveResult}, by first establishing tightness of the scaled processes, and then showing that every limit point is a solution of certain sub-differential equation that includes the averaging effect. In order to establish the latter, we had to go beyond the (already non-standard) approach of \cite{positiveResult}, and exploit the fact that limiting trajectories do exist in order establish the sub-differential equation in the boundary of the state space.
\end{itemize}
\end{enumerate}

\subsection{Previous work}
As previously mentioned, the first queueing theoretic analysis for a quantum switch was done in \cite{TowsleyStochAnalysis}. In this seminal paper, as well as in the follow-up work \cite{TowsleyTripartite,Towsley2020,Towsley2022}, the authors assume that there is always a request for connections to be made, and they use a continuous time Markov chain as a tractable approximation of the discrete-time nature of the quantum switches. Under these assumptions, they obtain throughput in a heterogeneous network, with/without decoherence and connections among two or three nodes.

In a slightly different line of work \cite{TowsleyMaxWeight,TowsleyProtocolDesign} that is the closest to this paper, the authors consider a discrete-time queueing model where requests for connections are indeed stored in a queue until enough entangled qubits are available to match with them. In particular, in \cite{TowsleyMaxWeight} they consider the case where any node can store only one qubit and where each qubit decoheres after a single time slot. In this setting, they show that a Max-Weight policy on the length of the requests queues is throughput optimal. On the other hand, in \cite{TowsleyProtocolDesign} they assume that there is an infinite quantum memory for qubits but that these do not decohere. In this case, they also show that Max-Weight is throughput optimal but with a weaker notion of stability that only involves the requests queues.
To the best of our knowledge, we are first to provide a theoretical analysis of a model with queueing on both the requests and the qubits and with a non-trivial distribution for the decoherence time. 

The matching queueing model used for quantum switches is akin to the models for assemble-to-order systems. However, assemble-to-order systems and their models rarely involve abandonments. Without abandonments, optimal policies often involve waiting for some time before performing a more informed matching (e.g., the so-called discrete-review policies \cite{discreteReview,ATO2006,dynamicControlMatching}). These policies are unsuitable for systems with abandonments, as they allow for unnecessary abandonments to occur. When incorporating abandonments, policies that wait before performing matchings can only be asymptotically optimal if the abandonment rate is asymptotically negligible \cite{matchingImpatientHeterogeneous,Blanchet2022}. Other recent works in the matching system that considers abandonments do so in continuous time and with with either deterministic patience \cite{perish_survey}, or with a mix of items with zero and infinite patience \cite{priorityAndImpatient,Gupta22}. In the latter, their optimal policies are also unsuitable for the case where all items can abandon, as they make the infinitely patient items wait unnecessarily.

Many other recent papers tackle the analysis and control of matching queueing systems. Some focus on finding general conditions for the stability \cite{Busic2013,hypergraphs} or instability \cite{instabilityMatching} of such systems. Others focus on throughput optimality and reward maximization in the long-term average \cite{Stolyar} or in a discounted setting \cite{optimalControl}. In particular, in \cite{ItaiItai2021} they show that greedy policies can be optimal in hindsight. Finally, some also focus on finding structural properties, such as showing that the steady-state queue length distribution has product form when the matching policy is First-Come-First-Served \cite{productForm} (even when there are abandonments \cite{francisco}), or that adding flexibility can hurt performance \cite{flexibilityHurts}. Finally, a recent work \cite{Jonckheere2022} considers a bipartite two-way matching queueing network with abandonments under a Max-Weight policy. While they obtain the stability region and other stochastic properties, they do not show optimality in any way. Moreover, none of these papers combine the challenges of having multi-way matchings together with (potentially correlated) abandonments that occur in the same time-scale as the arrivals, so their techniques cannot be directly translated to our setting.

\subsection{Notation}\label{sec:notation}
Given a real number $q$, we denote $q^-=-\min\big\{q,\, 0\big\}$ and $q^+=\max\big\{0,\, q\big\}$. Analogously, given a real-valued stochastic process $Q(\cdot)$, we denote $Q^-(\cdot)=-\min\big\{Q(\cdot),\, 0\big\}$ and $Q^+(\cdot)=\max\big\{0,\, Q(\cdot)\big\}$. We also denote $\mathbb{R}_+=[0,\infty)$ and $\mathbb{Z}_+=\{0,1,2,\dots\}$. Moreover, we denote
\[ D^2[0,T] := \left\{ f:[0,T]\to\mathbb{R}_+^2, \text{ c\`adl\`ag, and piece-wise constant} \right\}. \]
Finally, we denote as $\text{proj}_A(x)$ the projection of $x$ onto the set $A$, and we denote vectors with bold face fonts.

\section{Three-way matching in a Y-topology}\label{sec:Y}

In this section, we first introduce and analyze a two-way matching model (Subsection \ref{sec:twoWayMatching}). Then, we build upon this when we introduce and analyze our first model for a quantum switch consisting of a three-way matching network in a Y-topology (Subsection \ref{sec:Ymodel}).

\subsection{Two-way matching system}\label{sec:twoWayMatching}

\begin{wrapfigure}{R}{0.5\textwidth}
\centering
\begin{tikzpicture}[scale=0.7]
    \draw (4,0) -- (6,0);
    \draw (0,0) node {$\mu_1$};
    \draw [->] (0.5,0) -- (2,0);
    \draw [->] (3.5,0.5) -- (3.5,1) -- (2.5,1);
    \draw (2,1) node {$\gamma_1$};
    \draw (2,0.5) -- (4,0.5) -- (4,-0.5) -- (2,-0.5);
    \filldraw[fill=white] (5,0) circle (15pt);
    \draw (8,0.5) -- (6,0.5) -- (6,-0.5) -- (8,-0.5);
    \draw (10,0) node {$\mu_2$};
    \draw [->] (9.5,0) -- (8,0);
    \draw [->] (6.5,0.5) -- (6.5,1) -- (7.5,1);
    \draw (8,1) node {$\gamma_2$};
\end{tikzpicture}
\caption{System with two-way matchings and abandonments in both queues.}
\label{fig:twoSidedAbandonments}
\end{wrapfigure}
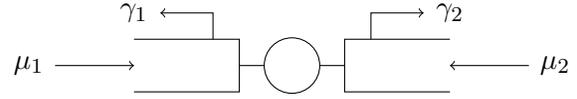
We consider a discrete-time system consisting of two FIFO queues with infinite buffers, where items arrive at each of the queues according to i.i.d. exogenous arrival processes denoted by $A_1(\cdot)$ and $A_2(\cdot)$. A pictorial representation is presented in Figure \ref{fig:twoSidedAbandonments}. We assume that the two arrival processes are independent of each other, with $\mathbb{E}[A_i(1)]=\mu_i$ and $\mathbb{E}\big[A_i(1)^2\big]<\infty$ for $i\in\{1,2\}$. Items of type $i$ depart either via immediate matching with the other type or via abandonment, where each item in the $i$-th queue abandons with probability $\gamma_i$ independent across different time slots. Within each time slot, we allow for an arbitrary correlation structure among the abandonments of the different items. We assume that the abandonments are realized at the beginning of the time slot, then the arrivals come in and join the queue, and finally, the items in the queues are matched with each other and depart the system. We also assume that the system is non-idling, which implies that at least one of the queues is empty at the end of any time slot. This system can be modeled as a discrete-time Markov chain $(Q_1(\cdot),Q_2(\cdot))$ over the state space $\mathbb{Z}_+\times\mathbb{Z}_+$ defined recursively by
\begin{align}
\label{eq: two_way_lindley_recursion}
    Q_i(t+1) &= Q_i(t) + A_i(t) - D_i(t) -M(t), \quad \text{ for } i \in \{1,2\},
\end{align}
for all $t\geq 0$, where
\[ M(t) = \min\limits_{i\in \{1,2\}}\big\{Q_i(t)-D_i(t)+A_i(t) \big\} \]
is the number of matchings in the $t$-th time slot, and
\[ D_i(t) = \sum\limits_{\ell=1}^{Q_i(t)} Z^{(i)}_\ell(t), \]
with $Z^{(i)}_\ell(t)\sim Ber(\gamma_i)$, denote the abandonments from the $i$-th queue. Here, for each $\ell\geq 1$, $\{Z^{(i)}_\ell(t):t\geq 0\}$ is a sequence of i.i.d. Bernoulli random variables with probability $\gamma_i$, and, for each $t\geq 0$, $\{Z^{(i)}_\ell(t):\ell\geq 1\}$ is a sequence of exchangeable Bernoulli random variables with probability $\gamma_i$. In particular, when the later are also i.i.d., we have $D_i(t)\sim Bin\big(Q_i(t),\gamma_i\big)$.

Note that $M(t)$ is defined in such a way that ensures that at least one of the queues is equal to zero at all times. Therefore, we can also model this system as the one-dimensional Markov chain $Q(\cdot) = Q_1(\cdot) - Q_2(\cdot).$
In particular, we have that $Q^+(\cdot) = Q_1(\cdot)$ and $Q^-(\cdot) = Q_2(\cdot)$, and that $Q(\cdot)$ can also be defined recursively by
\[ Q(t+1) = Q(t) + A_1(t) - A_2(t) - D_1(t) + D_2(t),\]
for all $t\geq 0$. In this setting, we define the \emph{throughput} as the rate of successful matchings, given by 
\[ \text{Throughput} := \liminf\limits_{t\to\infty} \frac{1}{t} \sum\limits_{k=1}^t M(k). \]
Furthermore, if the Markov chain $Q(\cdot)$ is ergodic, then the throughput is simply given by $\mathbb{E}[M(\cdot)]$, where the expectation is taken with respect to the steady-state distribution.

\subsubsection{No Abandonment}\label{sec:two_way_no_abandon}
When the abandonment probabilities for both queues are zero, that is, when $\gamma_1 = \gamma_2 =0$, the system acts as a one-dimensional random walk. Therefore, it is transient when $\mu_1\neq \mu_2$, and null recurrent when $\mu_1 = \mu_2$. However, in both cases, the throughput is equal, almost surely, to $\min\{\mu_1,\mu_2\}$. More details on this can be found in~\cite{menshikov2016non}.

\subsubsection{Abandonments on one side}\label{sec: two_way_one_side}

We now consider the case of the two-way matching system where there are no abandonments in one of the queues, that is, where $\gamma_1=0$. For ease of exposition, we denote $\mu_1=\lambda$, $\mu_2=\mu$, and $\gamma_2=\gamma$ in Figure \ref{fig:twoSidedAbandonments}. The queue on the right-hand side acts as a service queue, and the queue on the left-hand side acts as a request queue.


Although this is a server-less two-way matching system, it is closely related to the single-sever queue. Note that the  service queue is stable for any $\gamma \in (0,1]$. In the extreme case that $\gamma=1$, arrivals to the service queue that are not immediately matched are lost. Therefore, the request queue behaves exactly as a discrete-time GI/GI/1 queue where the services are the arrivals to the service queue. In particular, this implies that the request queue (and thus the whole system) is stable when $\lambda<\mu$. On the other hand, when $\gamma<1$, arrivals to the service queue that are not immediately matched stay in the queue for a geometrically distributed amount of time, or until they are matched. In relation to the single-server queue, this means that the unused ``services'' are stored for future use instead of being lost. Even though this does not exactly behave like a single server queue, the extra stored up services can only help in decreasing the queue length of the request queue. Therefore, the overall system is stable when $\lambda<\mu$. However, when $\lambda > \mu$, the storing of services are not enough to make the request queue stable. In Lemma \ref{lem:basicStability} given below, we formalize the above argument and establish that the stability and throughput of this system, which turn out to be the same as in a single-server queue, and do not depend on the value of $\gamma$.

\begin{lemma}\label{lem:basicStability}
For any $\gamma\in(0,1]$, we have the following.
    \begin{itemize}
        \item [(i)] If $\lambda < \mu$, then $Q(\cdot)$ is positive recurrent and the throughput is equal to $\lambda$.
        \item [(ii)] If $\lambda > \mu$, then $Q(\cdot)$ is transient and the throughput is equal to $\mu$.
    \end{itemize}
\end{lemma}

The positive recurrence when $\lambda < \mu$ can be shown by using a coupling argument, by considering a single server queue with `arrivals' as the arrivals to the request queue and `services' as the arrivals to the service queue. Then, the queue length in the request queue is always smaller than the queue length of the single server queue, and the result in Part (i) follows simply by using the stability of the single server queue. In Part (ii), we use a similar coupling to show that the overall queue length $Q(\cdot)$ is lower bounded by a transient random walk. 

Another way to prove Part (i) is to consider $L(x) = |x|$ as a Lyapunov function, and show that the drift of this Lyapunov function is negative outside a finite subset of the state space whenever $\lambda <\mu$. Then, the Foster-Lyapunov theorem \cite{meyn_tweedie_1992} implies system's stability. A necessary technicality to use the Foster-Lyapunov theorem is that the underlying Markov chain should be irreducible and aperiodic. This can be shown by using the fact that abandonments follow a binomial distribution. 
Mathematical details for this is provided in Appendix \ref{app:basicStability}. 

\begin{remark}
    While the request queue 
    requires $\lambda<\mu$ to be stable, the service queue is always stable thanks to the abandonments. However, as seen in Lemma \ref{lem:basicStability}, the magnitude of the abandonments plays no role in the request queue's stability region nor in the system's throughput.
\end{remark}

\subsubsection{Abandonments on both sides}\label{sec:bothSides}

We now consider the case where there are abandonments in both queues, that is, $\gamma_1,\gamma_2>0$. This system is depicted in Figure \ref{fig:twoSidedAbandonments}. In this case, the system's stability follows due  to abandonments on both sides. However, unlike the case where only one side has abandonments, in this case, throughput does depend on the value of $\gamma_1$ and $\gamma_2$. This is formalized in the Lemma \ref{lem:trivialStability} below.

\begin{lemma}\label{lem:trivialStability}
For any $\gamma_1,\gamma_2\in(0,1]$, the process $Q(\cdot)$ is positive recurrent with unique invariant distribution $\pi$, and the throughput is almost surely equal to
\begin{align}
    C_Y = \mathbb{E}_{\pi}[M] = \mu_1 - \gamma_1\mathbb{E}_\pi\big[Q^+\big] = \mu_2 - \gamma_2\mathbb{E}_\pi\big[Q^-\big], \label{eq:firstThr}
\end{align}
where
\[ M = \min\big\{ Q^+ - D_1(Q^+) + A_1(1),\, Q^- - D_2(Q^-) + A_2(1) \big\}. \]
\end{lemma}

Note that $M$ denotes the number of matchings for the two-dimensional representation of the two-way matching system. Thus, the first equality in Equation \eqref{eq:firstThr} follows simply by the definition of throughput and using the ergodicity and positive recurrence of the Markov chain. For second (or third) equality, one should note that $\mu_1$ (or $\mu_2$) is the average rate of arrivals to the first (or second) queue, and using the binomially distributed nature of the abandonments, the rate of abandonment in steady-state is $\gamma_1\mathbb{E}_\pi\big[Q^+\big]$ (or $\gamma_2\mathbb{E}_\pi\big[Q^-\big]$).  Finally, since the system is ergodic, $\mathbb{E}_{\pi}[M]$ is the average matching rate in the system. Since items can only leave the system via matching or abandonment, and as the rate of arrivals equals the rate of departures in steady-state, we have $\mu_1 = \gamma_1\mathbb{E}_\pi\big[Q^+\big] + \mathbb{E}_{\pi}[M]$. This can be formally shown by equating the drift of $Q_1(\cdot)$ (or $Q^+(\cdot)$) in Equation \eqref{eq: two_way_lindley_recursion} to zero in steady-state. The proof uses the Foster-Lyapunov theorem with $L(x) = |x|$ as the Lyapunov function once more, and it is given in Appendix \ref{app:trivialStability}.

\begin{remark}
    The results from subsections \ref{sec:two_way_no_abandon}, \ref{sec: two_way_one_side} and \ref{sec:bothSides} can be combined to say that, for any value of $\gamma_1,\gamma_2 \in[0,1]$, the throughput is given by $\mu_i - \gamma_i\mathbb{E}_{\pi}[Q_i]$ when the system is stable, and by $\min\{\mu_1,\mu_2\}$ when the system is unstable.
\end{remark}

\subsection{Model for three-way matching in a Y-topology}\label{sec:Ymodel}

We now consider a discrete-time system consisting of three FIFO queues with infinite buffers, where requests arrive to one of the queues as an i.i.d. process $A(\cdot)$, with $\mathbb{E}[A(1)]=\lambda$ and $\mathbb{E}\big[A(1)^2\big]<\infty$, and qubits are generated and join each one of the other two queues as i.i.d. processes $S_1(\cdot)$ and $S_2(\cdot)$, independent of each other, with $\mathbb{E}[S_i(1)]=\mu_i$ and $\mathbb{E}\big[S_i(1)^2\big]<\infty$ for $i\in\{1,2\}$. Qubits can depart from their respective queues at each time slot via abandonments with positive probabilities $\gamma_1$ and $\gamma_2$, respectively, with the same correlation structure as in the two-sided queue of the previous subsection. This system is depicted in Figure~\ref{fig:singleThreeWay}.

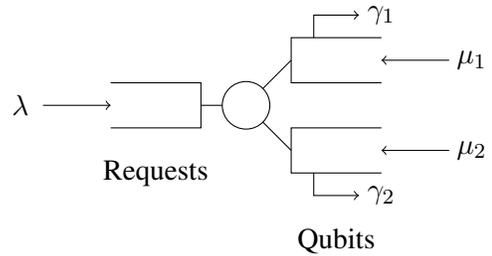
\begin{wrapfigure}{r}{0.45 \textwidth}
\centering
\begin{tikzpicture}[scale=0.6]
    \draw (5,0) -- (4,0);
    \draw (5,0) -- (6,1);
    \draw (5,0) -- (6,-1);
    
    \draw (3,-1.5) node {Requests};
    \draw (0,0) node {$\lambda$};
    \draw [->] (0.5,0) -- (2,0);
    \draw (2,0.5) -- (4,0.5) -- (4,-0.5) -- (2,-0.5);
    \filldraw[fill=white] (5,0) circle (15pt);
    
    \draw (8,1.5) -- (6,1.5) -- (6,0.5) -- (8,0.5);
    \draw (10,1) node {$\mu_1$};
    \draw [->] (9.5,1) -- (8,1);
    \draw[->] (6.5,1.5) -- (6.5,2) -- (7.5,2);
    \draw (8,2) node {$\gamma_1$};

    \draw (8,-0.5) -- (6,-0.5) -- (6,-1.5) -- (8,-1.5);
    \draw (10,-1) node {$\mu_2$};
    \draw [->] (9.5,-1) -- (8,-1);
    \draw[->] (6.5,-1.5) -- (6.5,-2) -- (7.5,-2);
    \draw (8,-2) node {$\gamma_2$};
    \draw (7,-3) node {Qubits};
\end{tikzpicture}
\setlength{\belowcaptionskip}{-10pt}
\caption{Three-way matching system in Y-topology.}
\label{fig:singleThreeWay}
\end{wrapfigure}

A request waiting in its queue can only depart from the system if it is matched with one qubit from each type. However, even if qubits are present in both qubit queues, the matching (request fulfillment) can fail. We assume that a request fulfillment is successful with probability $p>0$, and fails with probability $(1-p)$. In other words, if there is at least one item in each of the three queues, they depart immediately as a three-way matching with probability $p$, and as a two-way matching between the qubits with probability $1-p$ (with the request staying in its queue). If the request queue is empty, there is no two-way matching between the qubits. We also assume that the system is non-idling, in the sense that matching requests with pairs of qubits are attempted until one of the three queues becomes empty. Finally, we assume that there is no bound on the number of matches that can be made in any time slot.


We model this system as a discrete-time Markov chain $\big(N(\cdot), {\bf Q}(\cdot)\big)$ over the state space $\big\{(n,{\bf q})\in\mathbb{Z}_+\times\mathbb{Z}_+^2 : n.q_1.q_2 = 0\big\},$ where $N(\cdot)$ is the number of requests, and $Q_1(\cdot)$ and $Q_2(\cdot)$ are the number of qubits. Let $\mathbf Y(t)=\big\{Y_\ell(t): \ell,t\in\mathbb{Z}_+\big\}$ be a set of i.i.d. Bernoulli random variables with probability $p$, such that $Y_\ell(t)$ is the indicator that the $\ell$-th attempted three-way matching in the $t$-th time slot is successful. For $i\in\{1,2\}$, let $\big\{Z_\ell^{(i)}(t): \ell,t\in\mathbb{Z}_+\big\}$ be a set of exchangeable Bernoulli random variables with probability $\gamma_i$ such that, for each $\ell\in\mathbb{Z}_+$ and $i\in\{1,2\}$, the sequence $\big\{Z_\ell^{(i)}(t): t\in\mathbb{Z}_+\big\}$ is i.i.d.. Here $Z_\ell^{(i)}(t)$ is the indicator that the $\ell$-th qubit in the $i$-th queue in the $t$-th time slot abandons. Then, we 
have
\begin{align*}
    N(t+1) &= N(t) + A(t) - \sum_{\ell=1}^{M(t)} Y_\ell(t), && \text{ and }  &&
    Q_i(t+1) = Q_i(t) - D_i(t) + S_i(t) - M(t) 
\end{align*}
for $i\in\{1,2\}$, where
\[ D_i(t) = \sum\limits_{\ell=1}^{Q_i(t)} Z^{(i)}_\ell(t), \]
is the number of abandonments from the $i$-th qubits' queue, and $M(t)$ is the number of attempted three-way matchings. By the definition of $\mathbf Y(t)$, the term $\sum_{\ell=1}^{M(t)} Y_\ell(t)$ is a binomial random variable that represents the number of successful three-way matchings in the $t$-th time slot after $M(t)$ attempts. Further, by denoting 
\begin{equation}
    \label{eq: neg_bin}
    U(t) = \min\left\{m: \sum\limits_{\ell=1}^m Y_\ell(t) \geq N(t)+A(t) \right\},
\end{equation}
we have that the total number of three-way matchings attempted in the $t$-th slot is given by\
\[ M(t) = \min\left\{Q_1(t)-D_1(t)+S_1(t),\,\, Q_2(t)-D_2(t)+S_2(t),\,\, U(t) \right\}. \]
The term $U(t)$ in the definition of $M(t)$ is a negative binomial random variable that represents the number of three-way matching attempts needed to obtain $N(t)+A(t)$ successful ones. By the definition of $M(t)$, three-way matchings are attempted in each slot until there are either no more requests or any of the types of qubits. Thus, similar to that for the two-way matching system, $M(t)$ is defined to ensure that at least one of the queues is always empty. Finally, for this system, the throughput is defined as
\[ \text{Throughput} := \liminf\limits_{t\to\infty} \frac{1}{t} \sum\limits_{k=1}^t \sum\limits_{\ell=1}^{M(k)} Y_\ell(k). \]
Unlike in the definition of throughput of the two-way matching system, we now only consider the three-way matchings as part of the throughput, as two-way matchings in this system are failed attempts. If the Markov chain $Q(\cdot)$ is ergodic, then the throughput is simply be given by 
\[ \mathbb{E}\left[\sum_{\ell=1}^{M(\cdot)} Y_\ell(\cdot)\right] = p \mathbb{E}[M(\cdot)], \]
where the expectation is taken with respect to the steady-state distribution.

\subsubsection{Infinitely backlogged case}
We first consider a simplified case where the requests are infinitely backlogged, i.e., where $N(t)=\infty$ for all $t\geq 0$. In this case, the processes $Q_1(\cdot)$ and $Q_2(\cdot)$ behave the same as the two-way matching system of Subsection \ref{sec:bothSides}. That is, there exists a coupling such that, almost surely, we have $Q_1(t)=Q^+(t)$ and $Q_2(t)=Q^-(t)$ for all $t\geq 0$. Therefore, since $\gamma_1,\gamma_2>0$, Lemma \ref{lem:trivialStability} implies that the process ${\bf Q}(\cdot)$ is positive recurrent with invariant distribution $\pi$, and thus, $\mathbb{E}[M(\cdot)] = C_Y$. Therefore, since any matching is successful with probability $p$, the throughput of the Y-topology in the infinitely backlogged case is equal to $p\mathbb{E}[M(\cdot)] = p C_Y$.

\subsubsection{Non-infinitely backlogged case}
When the request queue $N(\cdot)$ is not infinitely backlogged, the stability of the system is no longer guaranteed by the abandonments. However, the stability region is closely related to the throughput in the infinitely backlogged case, as follows.

\begin{theorem}\label{thm:initialStability}
Let $C_Y$ be the throughput defined in Equation \eqref{eq:firstThr}. We have the following.
\begin{itemize}
\item [(i)] If $\lambda<p C_Y$, then $\big(N(\cdot),{\bf Q}(\cdot)\big)$ is positive recurrent, and the throughput is equal to $\lambda$.
\item [(ii)] If $\lambda>p C_Y$, then $\big(N(\cdot),{\bf Q}(\cdot)\big)$ is transient, and the throughput is equal to $p C_Y$.
\end{itemize}
\end{theorem}
In the previous subsection we observed that the rate of successful matchings is at most $p C_Y$. Thus, when the requests come at a rate higher than $p C_Y$, the system behaves like it is in the infinitely backlogged case, and the number of requests diverges with time. This gives us Theorem \ref{thm:initialStability} Part (ii). 

Proving the positive recurrence for the Y-topology in Part (i) is more challenging than proving the positive recurrence of the two-way matching system. Depending on the values of $\lambda$, $\mu_1$, and $\mu_2$, the drift of the request queue might be positive for infinitely many states. This means that a straightforward argument that uses $L(x) = |x|$ as a Lyapunov function (as in previous sections) does not work here. One might be able to find a suitable Lyapunov function and show that the one-step drift is negative if queue lengths are large, or one might prove stability using fluid limits, as we do for the W-topology in the upcoming section (Section \ref{sec:W}). However, we use a multi-step Lyapunov drift argument here to prove positive recurrence.

Note that, while the request queue is positive, the qubit queues behave like a two-way matching system, where the average rate of attempted matchings is $p C_Y$. Therefore, if $\lambda<p C_Y$, the drift of the request queue is negative during this time. To formally prove Part (i), we couple the system with two simpler ones.
\begin{enumerate}
    \item In one of them, qubits attempt two-way or three-way matchings irrespective of whether a request is present or not. In particular, if both qubit queues are non-empty, a matching between the qubits occurs, and they depart the system. Further, if a request is also present at that time, it is fulfilled with probability $p$. In this case, as qubits do not wait for a request, the request queue stochastically dominates the original request queue.
    \item In the other one, no matchings are ever attempted. As a result, the qubit queues behave as two $GI/M/\infty$ queues with potentially correlated services, which stochastically dominate the original qubit queues.
\end{enumerate} 
The main advantage of these two processes is that the evolution of the qubit queues does not depend on the state of the requests queue. Therefore, we use a multi-step Lyapunov argument on these new coupled processes to establish their positive recurrence. In particular, we use a number of steps $T$ large enough so that the multi-step drift for the request queue is close enough to $T(\lambda-p C_Y)$ when the queues are large enough. Finally, since the original processes were stochastically dominated by the new ones, the positive recurrence of the former follows. The complete proof for Theorem \ref{thm:initialStability} is given in Appendix \ref{app:initialStability}.

\begin{remark}
    Contrary to the two-way matching case, here the stability region of the system depends on the abandonment probabilities, as qubits of one type sometimes have to wait for the other type of qubits to match with a request. Therefore, qubits can be lost to abandonments even if there are outstanding requests. Moreover, the stability threshold is $p C_Y$, which is obtained as an expectation with respect to the steady-state distribution of a Markov chain. This is because while the requests queue is positive, it behaves as a single server queue with Markov modulated services.
\end{remark}

\section{Three-way matching in a W-topology}\label{sec:W}

Before going into the mathematical details of the W-topology, we motivate the model using a simple example. Suppose Alex and Jun go to a fast food restaurant. Alex orders a combo of items A and B, and Jun orders a combo of items B and C. Now, as this is a fast food chain, it tries to fulfill the first order that it can. It is simple to observe that, if items A and B are available before item C, Alex will be served first, even if Jun orders first. In fact, many customers might get served before Jun, depending on the time that it takes to prepare item C. Further, the restaurant might have to throw items away if they get too cold. These kinds of systems are called assemble-to-order systems \cite{ATO2006}, with the extra condition that inventories contain perishable items \cite{graves1982application}. In a quantum switch, we have qubits instead of items, and requests instead of customers. We present the model for the W-topology in terms of a quantum switch.

Consider a discrete-time system consisting of five FIFO queues with infinite buffers, where requests arrive to two of the queues as i.i.d. process $A_1(\cdot)$ and $A_2(\cdot)$, with $\mathbb{E}[A_i(1)]=\lambda_i$ and $\mathbb{E}\big[A_i(1)^2\big]<\infty$, and qubits are generated and join each one of the other three queues as i.i.d. processes $S_1(\cdot)$, $S_2(\cdot)$, and $S_3(\cdot)$ with $\mathbb{E}[S_i(1)]=\mu_i$ and $\mathbb{E}\big[S_i(1)^2\big]<\infty$. All these five processes are independent of each other. Moreover, each qubit departs from the system at each time slot via abandonments, with probability $\gamma_1$ from the first and third queues, and with probability $\gamma_2$ from the second queue. These abandonments are independent across different time slots, but they can be correlated among different qubits in the same time slot. This system is depicted in Figure \ref{fig:W}.

\begin{wrapfigure}{R}{0.5\textwidth}
\centering
\begin{tikzpicture}[scale=0.7]
    \draw (5,0) -- (4,0);
    \draw (5,0) -- (6,1);
    \draw (5,0) -- (6,-1);
    
    \draw (5,-2) -- (4,-2);
    \draw (5,-2) -- (6,-1);
    \draw (5,-2) -- (6,-3);
    
    \draw (0,0) node {$\lambda_1$};
    \draw [->] (0.5,0) -- (2,0);
    \draw (2,0.5) -- (4,0.5) -- (4,-0.5) -- (2,-0.5);
    \filldraw[fill=white] (5,0) circle (15pt);

    \draw (0,-2) node {$\lambda_2$};
    \draw [->] (0.5,-2) -- (2,-2);
    \draw (2,-1.5) -- (4,-1.5) -- (4,-2.5) -- (2,-2.5);
    \filldraw[fill=white] (5,-2) circle (15pt);

    \draw (3,-3.5) node {Requests};

    \draw (8,1.5) -- (6,1.5) -- (6,0.5) -- (8,0.5);
    \draw (10,1) node {$\mu_1$};
    \draw [->] (9.5,1) -- (8,1);
    \draw[->] (6.5,0.5) -- (6.5,0) -- (7.5,0);
    \draw (8,0) node {$\gamma_1$};

    \draw (8,-0.5) -- (6,-0.5) -- (6,-1.5) -- (8,-1.5);
    \draw (10,-1) node {$\mu_2$};
    \draw [->] (9.5,-1) -- (8,-1);
    \draw[->] (6.5,-1.5) -- (6.5,-2) -- (7.5,-2);
    \draw (8,-2) node {$\gamma_2$};

    \draw (8,-2.5) -- (6,-2.5) -- (6,-3.5) -- (8,-3.5);
    \draw (10,-3) node {$\mu_3$};
    \draw [->] (9.5,-3) -- (8,-3);
    \draw[->] (6.5,-3.5) -- (6.5,-4) -- (7.5,-4);
    \draw (8,-4) node {$\gamma_1$};
    
    \draw (7,-5) node {Qubits};
\end{tikzpicture}
\caption{Three-way matchings in a W-topology.}
\label{fig:W}
\end{wrapfigure}
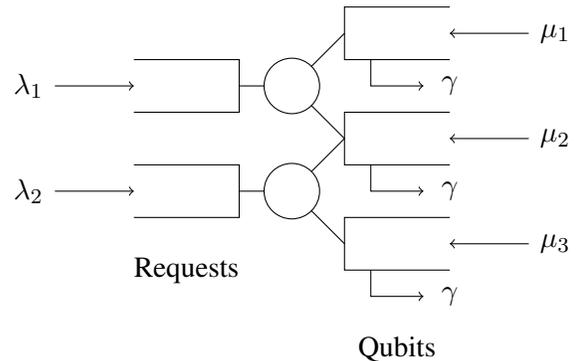

Similar to the Y-topology, we assume that a request fulfillment is successful with probability $p$ and fails with probability $(1-p)$. This means that if there is at least one item in each of three queues connected according to Figure \ref{fig:W}, they depart immediately as a three-way matching with probability $p$, or as a two-way matching between the qubits with probability $1-p$ (with the request staying in its queue). And if the request queue is empty, there is no matching between the qubits. Here we also assume that there is not bound on the number of matches that can be made in any time slot.

In the context of a quantum switch, this queueing system corresponds to the case where there are three nodes but only requests for connections between the first and second, and between the second and third. If a qubit for the second node is generated, and there are qubits for the other two nodes stored in the switch and outstanding requests for both types of connections (i.e., if there are more than one possible three-way matchings), a control policy (to be specified later) decides which connection to attempt first. We assume that the control policy is non-idling, that is, if there are qubits present to fulfill a request, then the system must attempt a three-way matching. The control policy only decides the priority between the requests queues. 

We model this system as a discrete-time Markov chain $\big({\bf N}(\cdot),{\bf Q}(\cdot)\big)$, where $N_1(\cdot)$ and $N_2(\cdot)$ are the number of requests, and $Q_1(\cdot)$, $Q_2(\cdot)$, and $Q_3(\cdot)$ are the number of qubits. Let $\mathbf Y_i(t)=\big\{Y^{(i)}_\ell(t): \ell,t\in\mathbb{Z}_+\big\}$ for $i\in\{1,2\}$ be two sets of i.i.d. Bernoulli random variables with probability $p$, such that $Y^{(i)}_\ell(t)$ is the indicator that the $\ell$-th attempted three-way matching of type $i$ at time $t$ is successful. Let $\big\{Z_\ell^{(i)}(t): \ell,t\in\mathbb{Z}_+,\, i\in\{1,3\}\big\}$ and $\big\{Z_\ell^{(2)}(t): \ell,t\in\mathbb{Z}_+\big\}$ be sets of exchangeable Bernoulli random variables with probabilities $\gamma_1$ and $\gamma_3$, respectively, such that, for each $\ell\in\mathbb{Z}_+$ and $i\in\{1,2,3\}$, the sequence $\big\{Z_\ell^{(i)}(t): t\in\mathbb{Z}_+\big\}$ is i.i.d.. Here $Z_\ell^{(i)}(t)$ is the indicator that the $\ell$-th qubit in the $i$-th queue in the $t$-th time slot abandons. We define the Markov chain recursively as follows.
\begin{align}
    N_1(t+1) &= N_1(t) + A_1(t) - \sum\limits_{\ell=1}^{M_1(t)} Y_\ell^{(1)}(t) \label{eq: w_state_vector_N1} \allowdisplaybreaks\\
    N_2(t+1) &= N_2(t) + A_2(t) - \sum\limits_{\ell=1}^{M_2(t)} Y_\ell^{(2)}(t) \label{eq: w_state_vector_N2}\allowdisplaybreaks\\
    Q_1(t+1) &= Q_1(t) - D_1(t) + S_1(t) - M_1(t)  \label{eq: w_state_vector_Q1}\allowdisplaybreaks\\
 Q_2(t+1) &= Q_2(t) - D_2(t) +S_2(t) - M_1(t) - M_2(t) \label{eq: w_state_vector_Q2}\allowdisplaybreaks\\
    Q_3(t+1) &= Q_3(t) - D_3(t) +S_3(t) - M_2(t) \label{eq: w_state_vector_Q3}
\end{align}
where
\[ D_i(t) = \sum\limits_{\ell=1}^{Q_i(t)} Z^{(i)}_\ell(t), \]
are the number of abandonments from the $i$-th qubits' queue, and $M_i(t)$ is the number of attempted matchings of type $i$ at time $t$. By definition of $\mathbf Y^{(i)}(t)$, the term $\sum_{\ell=1}^{M_i(t)} Y_\ell^{(i)}(t)$ is a binomial random variable that represents the number of successful three-way matchings in the $t$-th slots after $M_i(t)$ attempts. Further, by denoting 
\begin{equation}
\label{eq: w_neg_bin}
    U_i(t) = \min\left\{m: \sum\limits_{\ell=1}^m Y_\ell^{(i)}(t) \geq N_i(t)+A_i(t) \right\} \quad \text{for} \ i\in \{1,2\},
\end{equation}
we have that the number of two-way and three-way matchings of each type are given by
\begin{align}
    M_1(t) &= \min\big\{Q_1(t)-D_1(t)+S_1(t),\,\, Q_2(t)-D_2(t)+S_2(t) - X(t)M_2(t),\,\, U_1(t) \big\} \label{eq: w_matching_vector1}\\
    M_2(t) &= \min\big\{Q_2(t)-D_2(t)+S_2(t) - [1-X(t)]M_1(t),\,\, Q_3(t)-D_3(t)+S_3(t),\,\, U_2(t)\big\} \label{eq: w_matching_vector2}.
\end{align}
where $X(t)\in\{0,1\}$ is a (possibly randomized) control signal and the term $U_i(t)$ in the definition of $M_i(t)$ is similar as in Equation \eqref{eq: neg_bin}.
This definition implies that, in each time slot, three-way matchings are attempted until there are either no more requests, or no more of any of the types of qubits to fulfil any request. Thus, the state space of the Markov chain is $\left\{({\bf n},{\bf q})\in\mathbb{Z}_+^2\times\mathbb{Z}_+^3 : n_1.q_1.q_2 = 0, \,\, n_2.q_2.q_3 = 0\right\}.$
Finally, for the W-topology, the throughput of the $i$-th type of requests is defined as
\[ \text{Throughput}_i := \liminf\limits_{t\to\infty} \frac{1}{t} \sum\limits_{k=1}^t \sum\limits_{\ell=1}^{M_i(k)} Y^{(i)}_\ell(k). \]

The model for the W-topology can also be thought of as two three-way matching systems in a Y-topology stitched together. The process $(N_1(\cdot),Q_1(\cdot),Q_2(\cdot))$ represents the first Y-topology and the process $(N_2(\cdot),Q_2(\cdot),Q_3(\cdot))$ represents the second Y-topology, where the two Y-topologies share the resource in the qubit queue $Q_2(\cdot)$. If we consider one of the requests queues, say $N_1(\cdot)$, to be always empty ($\lambda_1 =0$), then the Y-topology corresponding to the other request queue, that is $(N_2(\cdot),Q_2(\cdot),Q_3(\cdot))$, behaves exactly like a three-way matching system in Section \ref{sec:Y}. Due to this relationship, the analysis of the W-topology is heavily based on the results in Section \ref{sec:Y}.

\begin{remark}
     The policies that can be implemented using the control signal $X(t)$ are ``non-idling'' in the sense that matchings are attempted whenever matchings are possible. In particular, if $X(t)=0$, all possible matchings amongst queues involving requests of type $1$ are attempted first, and then all possible matchings amongst queues involving requests of type $2$ are attempted with the remaining requests and qubits. If $X(t)=1$, the priority is reversed.
\end{remark}

\subsection{Infinitely backlogged case}
We start with the simplified case where both requests queues are infinitely backlogged, that is, where $N_1(t)=N_2(t)=\infty$, for all $t\geq 0$. In this case, the process ${\bf Q}(\cdot)$ is recursively defined as in equations \eqref{eq: w_state_vector_Q1}, \eqref{eq: w_state_vector_Q2} and \eqref{eq: w_state_vector_Q3}. Further, as $N_1(t)=N_2(t)=\infty$, we have that $U_1(t) = U_2(t) = \infty$ in Equation \eqref{eq: w_neg_bin}. Then the three-way matchings attempted are given by $M_1(t)$ and $M_2(t)$ in Equation \eqref{eq: w_matching_vector1} and Equation \eqref{eq: w_matching_vector2}, respectively, by substituting $U_1(t) = U_2(t) = \infty$.
In this case, we can model the system as a two-sided queue as in the following proposition.

\begin{proposition}\label{prop:equivalence}
Regardless of the control policy, the process $\tilde Q(\cdot) = Q_1(\cdot)+Q_3(\cdot) - Q_2(\cdot)$
behaves (path-wise) exactly as a two sided queue in Subsection \ref{sec:bothSides}, with arrivals $S_1(\cdot)+S_3(\cdot)$ and $S_2(\cdot)$, and abandonment probabilities equal to $\gamma_1$ and $\gamma_2$.
\end{proposition}

The proof of Proposition  \ref{prop:equivalence} (given in Appendix \ref{app:equivalence}) simply follows by rewriting the recursive equation for $\tilde Q(\cdot)$ and making the observation that, in infinitely backlogged case, $M_1(t)+M_2(t)$ does not depend on the value of the control signal $X(t)$. By the non-idling assumption, whenever $Q_2(\cdot)>0$, we have $Q_1(\cdot)+Q_3(\cdot)=0$, as otherwise two-way matchings are possible. Thus, we have $\tilde Q^+(\cdot) = Q_1(\cdot)+Q_3(\cdot)$ and $\tilde Q^-(\cdot) = Q_2(\cdot)$ with the condition $\tilde Q^+(\cdot).\tilde Q^-(\cdot) =0$. This describes the two-way matching system of Subsection \ref{sec:twoWayMatching}. The only difference here is in terms of the definition of the \textit{throughput}. In Subsection \ref{sec:twoWayMatching}, the throughput is defined as the average number of matchings, while here, we define throughput as the average number of `successful' matchings, with success probability $p>0$. Due to this, the throughput for $\tilde Q(\cdot)$ is equal to $p$ times the throughput of the two-way matching system in Subsection \ref{sec:twoWayMatching}. Using this equivalence, and the stability and throughput results in Lemma~\ref{lem:trivialStability}, we have the following.

\begin{corollary}
\label{cor: w_throughput_bound}
    For any non-idling control policy, the process $\tilde Q(\cdot)$ is positive recurrent, with unique invariant distribution $\pi_{1,2}$, and the total throughput of the system is
    \begin{equation}
    C_{1,2} := p\left(\mu_2 - \gamma_2 \mathbb{E}_{\pi_{1,2}}[ Q_2 ]\right) = p\left(\mu_1 + \mu_3 - \gamma_1 \mathbb{E}_{\pi_{1,2}}[ Q_1+Q_3 ]\right). \label{eq:thrBacklogged}
\end{equation}
Moreover, the individual throughputs are upper bounded by 
 \begin{align*}
    \overline{C}_{1} &:= p\left(\mu_1 - \gamma \mathbb{E}_{\pi^{(0)}_{1,2}}[ Q_1 ]\right) \qquad\text{ and }\qquad \overline{C}_{2} := p\left(\mu_3 - \gamma \mathbb{E}_{\pi^{(1)}_{1,2}}[ Q_3 ]\right),
 \end{align*}
and lower bounded by 
 \begin{align*}
     \underline{C}_{1} &:= C_{1,2} - \overline{C}_{2} = p\left(\mu_1 - \gamma \mathbb{E}_{\pi^{(1)}_{1,2}}[ Q_1 ]\right) \qquad\text{ and }\qquad \underline{C}_{2} := C_{1,2} - \overline{C}_{1} = p\left(\mu_3 - \gamma \mathbb{E}_{\pi^{(0)}_{1,2}}[ Q_3 ]\right),
 \end{align*}
where $\pi^{(0)}_{1,2}$ and $\pi^{(1)}_{1,2}$ are the steady-state distributions of the process ${\bf Q}(\cdot)$ when $X(t)=0$ and $X(t)=1$ for all $t\geq 0$, respectively.
\end{corollary}

Note that, Proposition \ref{prop:equivalence} and Lemma~\ref{lem:trivialStability} imply that the process $\tilde Q(\cdot)$ is positive recurrent. As for the throughputs, we make the following observations:
\begin{itemize}
    \item[-] Since the evolution of $\tilde Q(\cdot)$ does not depend on the control signal, the steady-state distributions of $Q_2(\cdot) = \tilde Q^-(\cdot)$ and $Q_1(\cdot)+Q_3(\cdot) = \tilde Q^+(\cdot)$ also do not depend on the control signal. Then, by the same argument as in Lemma \ref{lem:trivialStability}, Equation \eqref{eq:thrBacklogged} follows.
    \item[-] Even though the steady-state distribution of $Q_1(\cdot)+Q_3(\cdot)$ does does not depend on the control policy, the marginal distributions of $Q_1(\cdot)$ and $Q_3(\cdot)$ do. Among the non-idling policies, the throughput of $Q_1(\cdot)$ is maximized when it is given strict priority, that is when $X(t) = 0$ for all $t\geq 0$. This gives us the upper bound $\overline{C}_{1}$.
    \item[-] Similarly, the throughput for $Q_1(\cdot)$ is minimized when $Q_2(\cdot)$ is given strict priority, that is when $X(t) = 1$ for all $t\geq 0$. This gives us the lower bound $\underline{C}_{1}$.
    \item[-] Finally, as the total throughput remains the same for any policy, $\underline{C}_{1} + \overline{C}_{2} = \overline{C}_{1} + \underline{C}_{2} = C_{1,2}.$
\end{itemize}

\begin{remark}
Unsurprisingly, the upper and lower bounds for the individual throughputs are obtained when our control policy gives strict priority to one of them. However, the throughput of the non prioritized matching is not zero because our non-idling condition forces us to match low priority items when there are no high priority items to be matched, which is a frequent occurrence due to abandonments. 
\end{remark}

\subsection{Partially infinitely backlogged case}
Suppose that we have $N_1(t)=\infty$ and $N_2(t)<\infty$, for all $t\geq 0$. Since the first request queue is infinite, we set $X(t)=0$ for all $t\geq 0$, giving it strict priority. In this case, the process $(N_2(\cdot),{\bf Q}(\cdot))$ is defined recursively as in equations \eqref{eq: w_state_vector_N2}, \eqref{eq: w_state_vector_Q1}, \eqref{eq: w_state_vector_Q2}, and \eqref{eq: w_state_vector_Q3}. Further, as $N_1(t)=\infty$, we have that $U_1(t) = \infty$ in Equation \eqref{eq: w_neg_bin}. Then the three-way matchings attempted are given by $M_1(t)$ and $M_2(t)$ in Equation \eqref{eq: w_matching_vector1} and Equation \eqref{eq: w_matching_vector2}, respectively, by substituting $U_1(t) = \infty$.
Here, the second requests queue (and thus the whole system) can be unstable if its arrival rate is too large. Its stability region is provided in the following lemma.

\begin{lemma}\label{lem:halfStability}
    If $\lambda_2 < \underline{C}_{2}$, then the process $\big(N_2(\cdot),{\bf Q}(\cdot)\big)$ is positive recurrent, with unique invariant distribution $\pi_1$. Moreover, the throughputs of the first and second requests queues are $C_1(\lambda_2) := p\left(\mu_1 - \gamma_1 \mathbb{E}_{\pi_1}[Q_1]\right)$
    and $\lambda_2$, respectively.
\end{lemma}
This result can be established by using the same multi-step Lyapunov argument as for the single three-way matching system in the Y-topology (Theorem~\ref{thm:initialStability}). The only difference is that, instead of having the throughput of the Y-topology when it is infinitely backlogged ($p C_Y$), here we have the throughput of a completely backlogged W-topology with strict priority to the first type of requests ($\underline{C}_{2}$). Finally, the throughput is obtained using the same argument as in the proof of Lemma \ref{lem:trivialStability}.

\begin{remark}
It is easy to see that $C_1(\lambda_2)$ is a decreasing function of $\lambda_2$, and it is thus maximized when $\lambda_2=0$ (where we have $M_2(t)=0$ for all $t\geq 0$). In that case, the throughput of the first requests queue is equal to $C_1(0) := p\left(\mu_1 - \gamma_1 \mathbb{E}_{\overline{\pi}_1}[Q_1]\right),$
where $\overline{\pi}_1$ is the steady-state distribution of the system with $\lambda_2=0$. In particular, this is equal to $p$ times the throughput of the system in a Y-topology of Subsection \ref{sec:Ymodel} with $\lambda=\lambda_1$. Analogously, we can obtain $C_2(\lambda_1)$ and its upper bound $C_2(0)$.
\end{remark}

\subsection{Stability of the non infinitely backlogged case}

In this subsection we consider the case where none of the requests queues are infinitely backlogged. In this case, the control policy plays a central role in the stability of these queues. We consider the special case of the Max-Weight control policy with random tie-breakers, where the control signal is given by
\[ X(t) = \begin{cases}
0, & \text{ if } N_1(t)+A_1(t) > N_2(t)+A_2(t) \\
Ber(1/2), & \text{ if } N_1(t)+A_1(t) = N_2(t)+A_2(t) \\
1, & \text{ if } N_1(t)+A_1(t) < N_2(t)+A_2(t).
\end{cases}
\]
Note that our Max-Weight policy compares the size of the requests queues and ignores the qubits' queues. This is because, making an analogy with a server-based queueing system, qubits play the role of ``services'', while the requests play the role of ``jobs''. Furthermore, abandonments make the queues of qubits always stable, while only matching with qubits can stabilize the queues of requests. Max-Weight scheduling was also used in  \cite{TowsleyMaxWeight,TowsleyProtocolDesign}.

Recall that, in the Y-topology of Section \ref{sec:Y}, the throughput in the infinitely backlogged case determined the stability region of the non-infinitely backlogged case. In particular, the stability region was the set of all arrival rates lower than the throughput in the infinitely backlogged case. Analogously, here in the W-topology, the throughputs obtained in the either completely or partially infinitely backlogged case serve as upper bounds for the individual arrival rates in the stability region. With this in mind, we expect the stability region of the Max-Weight policy to roughly look like the green shaded part in Figure \ref{fig:stabilityRegion}. Moreover, since the upper bounds for the throughputs are independent of the control policy, we expect any non-idling policy to be unstable outside the green shaded part in Figure \ref{fig:stabilityRegion}. This is formalized in Theorem \ref{thm:stability}.

\begin{figure}[ht!]
    \centering
    \begin{tikzpicture}[yscale=1,xscale=2]
        \filldraw[color=green!30] (0,0) -- (0,4) to[out=-5,in=155] (1.5,3.5) -- (3.5,1.5) to[out=-60,in=100] (4,0) -- (0,0);

        \draw [->] (-0.5,0) -- (4.67,0);
        \draw [->] (0,-0.5) -- (0,5);
        \draw (4.87,0) node {$\lambda_1$};
        \draw (0,5.3) node {$\lambda_2$};

        \draw (0,4) to[out=-10,in=155] (1.5,3.5) -- (3.5,1.5) to[out=-60,in=100] (4,0);


        \draw[dotted, thick] (1.5,3.5) -- (1.5,0);
        \draw[dotted, thick] (1.5,3.5) -- (0,3.5);
        \draw[dotted, thick] (3.5,1.5) -- (3.5,0);
        \draw[dotted, thick] (3.5,1.5) -- (0,1.5);

        \draw (1.5,-0.3) node {$\underline{C}_1$};
        \draw (3.5,-0.3) node {$\overline{C}_1$};
        \draw (4,-0.3) node {$C_1(0)$};

        \draw (-0.2,1.5) node {$\underline{C}_2$};
        \draw (-0.2,3.5) node {$\overline{C}_2$};
        \draw (-0.25,4) node {$C_2(0)$};
        
        \draw (1,4.1) node {$\lambda_2=C_2(\lambda_1)$};
        \draw (4.3,1) node {$\lambda_1=C_1(\lambda_2)$};

        \draw (3,2.7) node {$\lambda_1+\lambda_2 = C_{1,2}$};
    \end{tikzpicture}
    \caption{Max-Weight is stable in the green shaded area, and all non-idling policies are unstable outside of it.}
    \label{fig:stabilityRegion}
\end{figure}
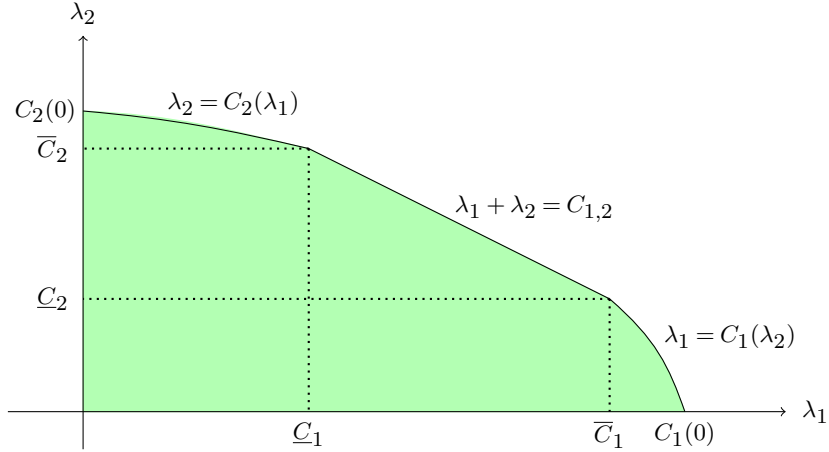

\begin{theorem}\label{thm:stability}
    We have the following.
    \begin{itemize}
        \item[(i)] If $\lambda_1+\lambda_2<C_{1,2}$, $\lambda_1<C_1(\lambda_2)$, and $\lambda_2<C_2(\lambda_1)$ (green shaded region in Figure \ref{fig:stabilityRegion}), then the process $\big({\bf N}(\cdot),{\bf Q}(\cdot)\big)$ is positive recurrent under the Max-Weight policy.
        \item[(ii)] If $\lambda_1+\lambda_2>C_{1,2}$, $\lambda_1>C_1(\lambda_2)$, or $\lambda_2>C_2(\lambda_1)$ (non shaded region in Figure \ref{fig:stabilityRegion}), then the process $\big({\bf N}(\cdot),{\bf Q}(\cdot)\big)$ is transient under any matching policy.
    \end{itemize}
\end{theorem}
In particular, this theorem implies that Max-Weight is throughput optimal, as it is under the simplifying assumptions of \cite{TowsleyMaxWeight,TowsleyProtocolDesign}.

\begin{remark}
    In order to prove Theorem~\ref{thm:stability}, we use the stability of the Y-topology given by Theorem~\ref{thm:initialStability}. Therefore, we cannot obtain Theorem~\ref{thm:initialStability} as a special case of Theorem~\ref{thm:stability} (when $\lambda_2=0$), as this would create a circular argument.
\end{remark}

\begin{remark}
    It can be shown that the capacity region defined in Theorem~\ref{thm:stability} is convex by exploiting the same fluid limit approach as for this theorem. Indeed, given two arrival rates for which Max-Weight is stable, we can construct a policy that alternates between the stochastic behaviors of Max-Weight for each stable arrival rate (and use the fluid limit approach to prove its stability).
\end{remark}

\begin{proof}{Proof sketch of Theorem \ref{thm:stability}:}
To establish the positive recurrence of $\big({\bf N}(\cdot),\,{\bf Q}(\cdot)\big)$ under Max-Weight policy in its stability region, we use a fluid limit approach. This is obtained in six steps, with some intermediate results being of independent interest and stated later as standalone results.
\begin{enumerate}
    \item {\bf Re-scaling the process:} The first step is to define the scaling for the processes ${\bf N}(\cdot)$ and ${\bf Q}(\cdot)$ under which we take the fluid limit. Since our objective is to prove stability, we accelerate time by some scaling factor $k$ and scale down space by the same scaling factor $k$. To obtain a meaningful limit, the initial condition is also scaled up by $k$. The fluid limit for single server queues typically uses this scaling. Still, it is not appropriate for infinite server queues, as the sample paths would collapse to zero regardless of the initial condition. In our setting, since the qubit queues have abandonments (which gives them a similar behavior to infinite server queues), the re-scaled process of qubits collapses to zero. Therefore, we focus on the fluid limit of ${\bf N}(\cdot)$.
    \item {\bf Tightness of sample paths:} If the arrival processes to all queues were uniformly bounded, then a linear interpolation of the discrete time processes would yield uniformly Lipschitz functions. In that case, the tightness of sample paths would follow immediately by Arzel\`a-Ascoli's theorem. However, since we only have a bound on the second moments of the arrival processes, we use the fact that the sample paths are approximately uniformly Lipschitz to show that there exists a measurable set $\mathcal{C}$, with $\mathbb{P}(\mathcal{C}) = 1$, such that for all $\omega\in\mathcal{C}$, any sequence of sample paths of the re-scaled processes $\left\{{\bf N}^{(k)}(\omega,\cdot)\right\}_{k=1}^\infty$ contains a further subsequence that converges uniformly over compacts to a uniformly Lipschitz trajectory ${\bf n}(\cdot)$, as $k\to\infty$. In particular, this argument relies on the Lipschitz continuity of the limiting trajectories and does not need further regularity assumptions on the drift. Therefore, it allows us to obtain tightness even when the drift is discontinuous.
    \item {\bf Derivatives of the limit points:} We characterize the derivative of any limit point ${\bf n}(\cdot)$ at any regular point, and show that it is identical to the drift of a certain differential equation (cf. Definition \ref{def:fluid}). If $n_1(t)>0$ and $n_2(t)>0$, the derivatives are characterized by obtaining bounds on how much the process can change in an $\epsilon$ amount of (re-scaled) time, and taking the limit as $k\to\infty$. This gives us bounds on how much the fluid solution changes in an $\epsilon$ amount of time, which reflect the averaging effect of the qubit processes. Then, we divide this change by $\epsilon$ and take the limit as $\epsilon\to 0$ to obtain the derivatives. If either $n_1(t)=0$ or $n_2(t)=0$, the derivatives are obtained by exploiting the previous case, and the continuity of derivatives at the regular times. In particular, we obtain that the limit point ${\bf n}(\cdot)$ must be a solution of the differential equation, yielding also, as a corollary, the existence of its solutions.
    \item {\bf Uniqueness of limit points:} Since the drift is not continuous, the uniqueness of limit points cannot be obtained by Picard's theorem for differential equations. Therefore, this uniqueness needs to be established in an ad-hoc way, by concatenating a finite number of trajectories that are provably unique by Picard's Theorem.
    \item {\bf Stability of the fluid solutions:} We exploit the structure of the derivatives of the fluid limits to show that all fluid trajectories converge to $0$ in finite time. In particular, we show that this time is upper bounded by an affine function of the norm of the initial condition. This finite-time convergence is crucial for the purposes of obtaining positive recurrence of the original processes.
    \item {\bf From fluid stability to stochastic stability:} Finally, we obtain the positive recurrence of the original Markov chain $\big({\bf N}(\cdot),\,{\bf Q}(\cdot)\big)$ by combining the finite time convergence of the fluid limit of ${\bf N}(\cdot)$ to zero with the collapse of ${\bf Q}(\cdot)$ to zero in the limit, and applying known results~\cite{DaiHarrisonBook}.
\end{enumerate}
Step 1 is explained in more detail in Subsection \ref{sec:rescaling}. Steps 2-4 establish the almost sure convergence of the re-scaled processes $\left\{{\bf N}^{(k)}(\omega,\cdot) \right\}_{k=1}^\infty$, uniformly over compacts, to their fluid limit. This is stated as a standalone result in Theorem \ref{thm:fluid}. The stability of the fluid limits implied by Step 5 is stated in Theorem \ref{thm:uniquenessAndStability}. Finally, the details of Step 6 and the proof of the second part of Theorem \ref{thm:stability} are given in Appendix \ref{app:stability}.
\hfill $\square$\end{proof}

\subsection{Transient analysis}
In this subsection, we will study the transient behavior of the system through a fluid approximation. Moreover, this approximation is also used to establish the positive recurrence of the system when none of the requests queues are infinitely backlogged.

\subsubsection{Fluid approximation}
We first define the (deterministic) fluid model for our system.

\begin{definition}[Fluid model]\label{def:fluid}
    Given an initial condition ${\bf n^0}\in\mathbb{R}_+^2$ with $\left\|{\bf n^0}\right\|_1>0$, a continuous function ${\bf n}:[0,\infty)\to \mathbb{R}_+^2$ is said to be a solution to the fluid model (or fluid solution) if:
    \begin{itemize}
        \item [(i)] ${\bf n}(0)={\bf n^0}$.
        \item [(ii)] For all $t\geq 0$, outside of a set of Lebesgue measure zero, ${\bf n}(t)$ is differentiable and satisfies
        \begin{align*}
    \frac{dn_1(t)}{dt} = \lambda_1 - R_1({\bf n}(t)), \ \text{ and } \ \frac{dn_2(t)}{dt} = \lambda_2 - R_2({\bf n}(t)),
\end{align*}
where $R_i({\bf n})$ is as follows:
\begin{itemize}
    \item[(a)] If $n_i>n_j>0$, then $R_i({\bf n}) = \overline{C}_i \ \text{ and } \ R_j({\bf n}) = \underline{C}_j.$
    \item[(b)] If $n_i=n_j>0$, then
    \begin{align*}
        R_i({\bf n}) = \text{proj}_{\big[\underline{C}_i,\, \overline{C}_i\big]} \left( \frac{C_{1,2}+\lambda_i-\lambda_j}{2} \right) \ \text{ and } \
        R_j({\bf n}) = \text{proj}_{\big[\underline{C}_j,\, \overline{C}_j\big]} \left( \frac{C_{1,2}+\lambda_j-\lambda_i}{2} \right).
    \end{align*}
    \item[(c)] If $n_i>n_j=0$, then 
    \begin{align*}
        R_i({\bf n}) = \overline{C}_i \mathds{1}_{\{\lambda_j \geq \underline{C}_j\}} + C_i(\lambda_j) \mathds{1}_{\{\lambda_j < \underline{C}_j\}} \ \text{ and } \
        R_j({\bf n}) = \underline{C}_j \mathds{1}_{\{\lambda_j \geq \underline{C}_j\}} + \lambda_j \mathds{1}_{\{\lambda_j < \underline{C}_j\}}.
    \end{align*}
    \item[(d)] If $n_i=n_j=0$, then $R_i({\bf n}) = \lambda_i \ \text{ and } \
        R_j({\bf n}) = \lambda_j.$
\end{itemize}
\end{itemize}
\end{definition}

These rates have a very close connection with the throughputs in the infinitely backlogged systems, which we explore below.
\begin{itemize}
    \item[(a)] If $n_i(t)>n_j(t)>0$, then both queues will remain positive and the $i$-th queue will remain the largest one for some time after time $t$. As a result, both requests queues will behave as if they were infinitely backlogged during that period, while giving full priority to the $i$-th queue. Thus, the fluid rates $R_i({\bf n}(t))$ and $R_j({\bf n}(t))$ are exactly the throughputs of that case.
    \item[(b)] If $n_i(t)=n_j(t)>0$, then both queues will remain positive for some time after time $t$. As a result, both requests queues will behave as if they were infinitely backlogged. Since we have $n_1(t)=n_2(t)$ and we are using Max-Weight, the control signal gives priority to one or the other queue with the goal of equalizing them. In particular, the drifts of the queues are equalized if
    \begin{align*}
        R_i({\bf n}(t)) = \frac{C_{1,2}+\lambda_i-\lambda_j}{2} \ \text{ and } \
        R_j({\bf n}(t)) = \frac{C_{1,2}+\lambda_j-\lambda_i}{2}.
    \end{align*}
    However, the individual throughputs are lower and upper bounded by constants, so these rates might be infeasible. Hence, the fluid rates will be the projection of these ideal rates onto the set of possible rates.
    \item[(c)] If $n_i(t)>n_j(t)=0$, then the $i$-th queue will remain positive and will remain the largest one for some time after time $t$. Since the other request's queue length is zero (at least at the fluid scale), the fluid model will behave as the partially infinitely backlogged system. If $\lambda_j<\underline{C}_j$, then this system is stable, and the fluid rates are equal to the throughputs in that case. However, if $\lambda_j > \underline{C}_j$, then the partially backlogged system is not stable, and the smaller queue will become positive (and hence have the same fluid rate as if both were positive). In the case $\lambda_j = \underline{C}_j$ we also have $\overline{C}_i=C_i(\lambda_j)$, and the fluid rate follows by continuity of the rates as a function of ${\bf \lambda}$.
    \item[(d)] Since $\|{\bf n}(0)\|_1>0$, then the case $n_i(t)=n_j(t)=0$ can only occur for some $t>0$, and only if the overall system is stable. Hence, the fluid rates are what make ${\bf n}$ stay at zero.
\end{itemize}

\begin{remark}
As it was the case for the throughputs in the infinitely backlogged system, when $n_1,n_2>0$, it can be checked that we have $R_i({\bf n})\in \big[\underline{C}_i,\, \overline{C}_i\big]$ and $R_1({\bf n})+R_2({\bf n})= C_{1,2}$.
This means that the fluid rates are in a segment within a $2$-dimensional space, drawn in red in Figure \ref{fig:rateCases}. In the parlance of the Max-Weight policies, this means that the possible ``schedules'' for this policy lie on this segment when both queues are positive.
\end{remark}

\subsubsection{Convergence of sample paths to a fluid solution}\label{sec:rescaling}
We now formally justify the use of the fluid model as an approximation of the stochastic system. In order to obtain a fluid approximation, we accelerate time and contract the state by some common scaling factor $k$. That is, we consider the sequence of scaled processes $\left\{\left({\bf N}^{(k)}(\cdot),{\bf Q}^{(k)}(\cdot)\right)\right\}_{k=1}^\infty$ such that
\[ {\bf N}^{(k)}(t) = \frac{{\bf N}(kt)}{k} \qquad \text{and} \qquad {\bf Q}^{(k)}(t) = \frac{{\bf Q}(kt)}{k}, \]
for all $t\geq 0$ and $k\geq 1$. Moreover, to obtain a meaningful limit, the initial condition also needs to be of order $k$.

Since time is accelerated by a factor of $k$, then the scaled abandonment rate is initially $\mathcal{O}\big(k^2\big)$, emptying the queues in $o(1)$ units of (re-scaled) time. Formally, since the $Q_i(\cdot)$ processes are upper bounded by GI/M/$\infty$ queues with potentially correlated services (where qubits only depart due to abandonments), we have
\begin{align}
    \lim\limits_{k\to\infty} \frac{1}{k} \mathbb{E}\Big[ Q_i(kt) \Big] &\leq \lim\limits_{k\to\infty} \frac{1}{k} \left[ \left( kQ_i(0) - \frac{\mu_i}{\gamma_i}\right) (1-\gamma_i)^{kt} + \frac{\mu_i}{\gamma_i} \right] = 0, \label{eq:collapseOfQ}
\end{align}
for all $t>0$. Therefore, for the fluid limit, we only consider the case where $\left\|{\bf Q}^{(k)}(0)\right\|_1=0$, for all $k\geq 1$. In this case, we have the following sample path convergence result.

\begin{theorem}[Convergence of sample paths]\label{thm:fluid}
    Fix $T>0$, and ${\bf n^0}\in\mathbb{R}_+^2$ with $\left\|{\bf n^0}\right\|_1>0$. If
    \[ \left\|{\bf Q}^{(k)}(0)\right\|_1=0, \quad \forall\,k\geq 1, \quad \text{ and } \quad \lim\limits_{k\to\infty} \left\| {\bf N}^{(k)}(0) - {\bf n^0} \right\|_1 = 0, \quad a.s., \]
    then
    \[ \lim\limits_{k\to\infty}\sup\limits_{t\in[0,T]} \left\| {\bf N}^{(k)}(t) - {\bf n}(t) \right\|_1 = 0, \quad a.s., \]
    where ${\bf n}(\cdot)$ is a fluid solution with initial condition ${\bf n^0}$. Furthermore, such solution is unique.
\end{theorem}

\begin{remark}
    In light of this theorem, a solution to the fluid model, ${\bf n}(\cdot)$, can be thought of as a deterministic approximation to the sample paths of the stochastic process ${\bf N}(\cdot)$ when all coordinates of ${\bf N}(0)$ are very large and time is accelerated proportionally to $\|{\bf N}(0)\|_1$. Note that the fluid model does not include a variable associated with the different number of qubits ${\bf Q}(\cdot)$. This is because the qubits process ${\bf Q}(\cdot)$ evolves on a faster time scale than the requests process ${\bf N}(\cdot)$. We thus have a process with two different time scales: on the one hand, the qubits' queues evolve on a fast time scale, and from its perspective, the requests process ${\bf N}(\cdot)$ appears static; on the other hand, the requests process ${\bf N}(\cdot)$ evolves on a slower time scale and from its perspective, the qubits' queues appear to be at stochastic equilibrium. This latter property is manifested in the drift of the fluid model through $R_1({\bf n})$ and $R_2({\bf n})$, which can be interpreted as the long-term average three-way matching rates of the first and second types of requests, respectively, when the requests processes are fixed at the state ${\bf n}$.
\end{remark}
On the technical side, the proof is challenging for two reasons:
\begin{enumerate}
    \item {\bf Discontinuity of the drift:} Since the drift is discontinuous, and in particular not Lipschitz continuous, we cannot apply well-known fluid limit theorems that rely the Lipschitz continuity of the drift (such as Kurtz's theorem for density dependent Markov processes). Furthermore, this discontinuity also prevents us from applying Picard's theorem to establish the existence and uniqueness of fluid solutions, which need to be proved in an ad-hoc way.
    \item {\bf Infinite-state averaging effect:} Since the ${\bf Q}(\cdot)$ process evolves on a faster time scale than ${\bf N}(\cdot)$, the evolution of the latter at the fluid scale depends on an averaged version of ${\bf Q}(\cdot)$, which is also non-standard. Moreover, the fact that the ${\bf Q}(\cdot)$ process has an infinite state space further complicates the proof.
\end{enumerate}
In order to overcome these two challenges, we prove this theorem by adapting the approach used in \cite{positiveResult}. The complete proof is given in Appendix \ref{app:fluid}.

\subsubsection{Transient analysis of the fluid solutions}

We now analyze the transient behavior of the fluid solutions. In particular we show that, if the arrival rates are small enough, then fluid solutions converge to zero in finite time. This is formalized in the following statement.

\begin{theorem}[Stability of fluid solutions]\label{thm:uniquenessAndStability}
    Let ${\bf n}(t)$ be the unique fluid solution with initial condition ${\bf n^0}$. If $\lambda_1+\lambda_2<C_{1,2}$, $\lambda_1<C_1(\lambda_2)$, and $\lambda_2<C_2(\lambda_1)$, then there exists a constant $B$ that does not depend on the initial condition ${\bf n^0}$ such that
    $ \|{\bf n}(t)\|_1=0$ for all $t\geq B \left\|{\bf n^0}\right\|_1. $
\end{theorem}
The proof involves exploiting the fact that the drift is piece-wise constant to obtain bounds on the times when the first and second coordinates of ${\bf n}(\cdot)$ hit zero, and it is given in Appendix \ref{app:uniquenessAndStability}.

When the assumptions of Theorem \ref{thm:uniquenessAndStability} hold (i.e., when the system is stable under Max-Weight), there are three qualitatively different cases for the transient behavior of the fluid solutions, depending on the values of $\lambda_1$ and $\lambda_2$. These are depicted in Figure \ref{fig:rateCases} in the shaded green (Case 1), turquoise (Case 2), and blue (Case 3) areas.

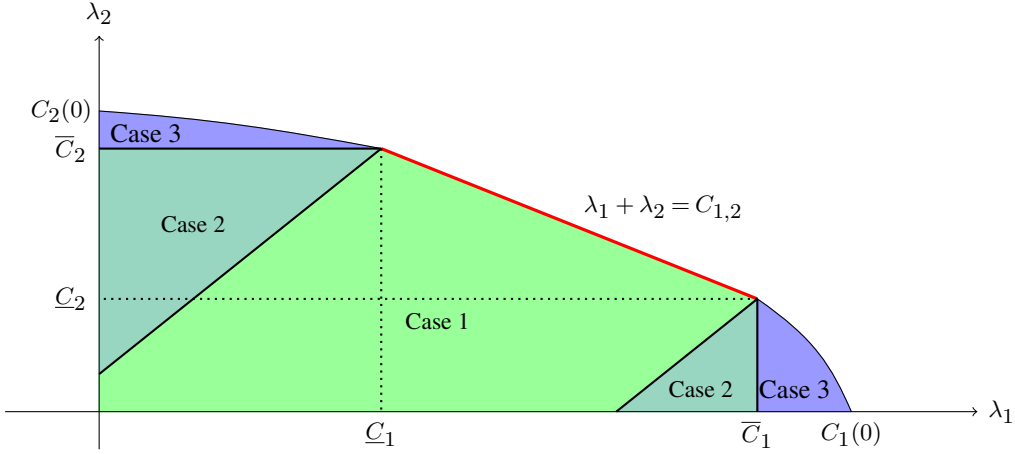
\begin{figure}[ht!]
    \centering
    \begin{tikzpicture}[yscale=1,xscale=2.5]
        

        \filldraw[color=green!40] (0,0) -- (0,0.5) -- (1.5,3.5) -- (3.5,1.5) -- (2.75,0) -- (0,0);
        \filldraw[color=blue!40!green!40] (0,0.5) -- (1.5,3.5) -- (0,3.5) -- (0,0);
        \filldraw[color=blue!40!green!40] (2.75,0) -- (3.5,0) -- (3.5,1.5) -- (2.75,0);
        \filldraw[color=blue!40] (3.5,0) -- (3.5,1.5) to[out=-60,in=100] (4,0) -- (3.5,0);
        \filldraw[color=blue!40] (0,3.5) -- (0,4) to[out=-10,in=155] (1.5,3.5) -- (0,3.5);

        \draw [->] (-0.5,0) -- (4.67,0);
        \draw [->] (0,-0.5) -- (0,5);
        \draw (4.8,0) node {$\lambda_1$};
        \draw (0,5.3) node {$\lambda_2$};

        \draw (0,4) to[out=-10,in=155] (1.5,3.5) -- (3.5,1.5) to[out=-60,in=100] (4,0);
        \draw[very thick, color=red] (1.5,3.5) -- (3.5,1.5);

        \draw[thick] (2.75,0) -- (3.5,1.5);
        \draw[thick] (0,0.5) -- (1.5,3.5);

        \draw[thick] (1.5,3.5) -- (0,3.5);
        \draw[thick] (3.5,1.5) -- (3.5,0);
        \draw[dotted, thick] (3.5,1.5) -- (0,1.5);
        \draw[dotted, thick] (1.5,3.5) -- (1.5,0);

        \draw (1.5,-0.3) node {$\underline{C}_1$};
        \draw (3.5,-0.3) node {$\overline{C}_1$};
        \draw (4,-0.3) node {$C_1(0)$};

        \draw (-0.15,1.5) node {$\underline{C}_2$};
        \draw (-0.15,3.5) node {$\overline{C}_2$};
        \draw (-0.2,4) node {$C_2(0)$};

        \draw (3,2.7) node {$\lambda_1+\lambda_2 = C_{1,2}$};
        
        \draw (1.8,1.2) node {Case 1};
        
        \draw (0.5,2.5) node {Case 2};
        \draw (3.2,0.3) node {Case 2};
        
        \draw (0.25,3.7) node {\small Case 3};
        \draw (3.7,0.3) node {\small Case 3};
    \end{tikzpicture}
    \caption{Pictorial representation of the three different cases.}
    \label{fig:rateCases}
\end{figure}

\noindent{\bf Case 1:} Suppose that the assumptions of Theorem \ref{thm:uniquenessAndStability} hold, and that $\lambda_1 + \lambda_2 - C_{1,2} \geq 2\max\left\{ \lambda_1 - \overline{C}_1,\,\, \lambda_2 - \overline{C}_2 \right\}.$
This corresponds to the case where the arrival rates $\lambda_1$ and $\lambda_2$ are relatively close to each other (area shaded in green in Figure \ref{fig:rateCases}), and implies
\begin{align*}
    \frac{\lambda_1+\lambda_2-C_{1,2}}{2} \geq \lambda_1 - \overline{C}_1, \quad \text{ and } \quad \frac{\lambda_1+\lambda_2-C_{1,2}}{2} \geq \lambda_2 - \overline{C}_2.
\end{align*}
        
        
        

\begin{figure}
    \centering
    \begin{subfigure}{.45\textwidth}
    \hspace{-35pt} \includegraphics[scale = 0.3]{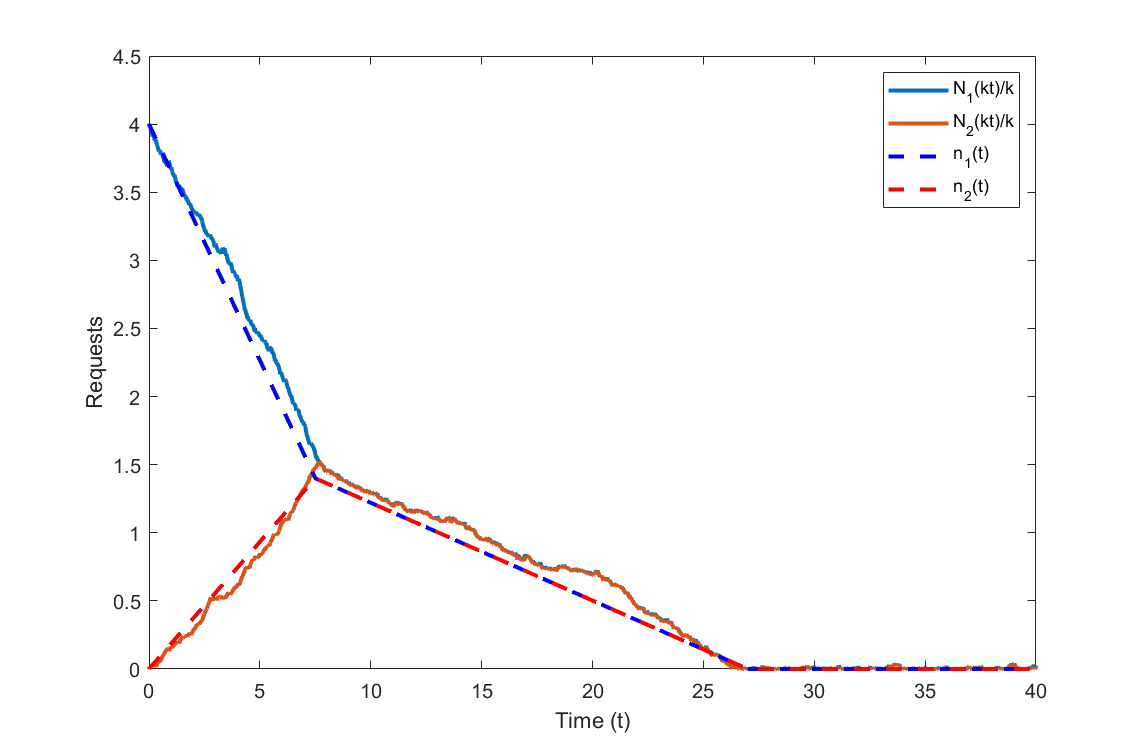}
    \caption{Case 1}
    \label{fig:case1}
    \end{subfigure}
    \hspace{-10pt} \begin{subfigure}{.45\textwidth}
        \centering
    \includegraphics[scale=0.3]{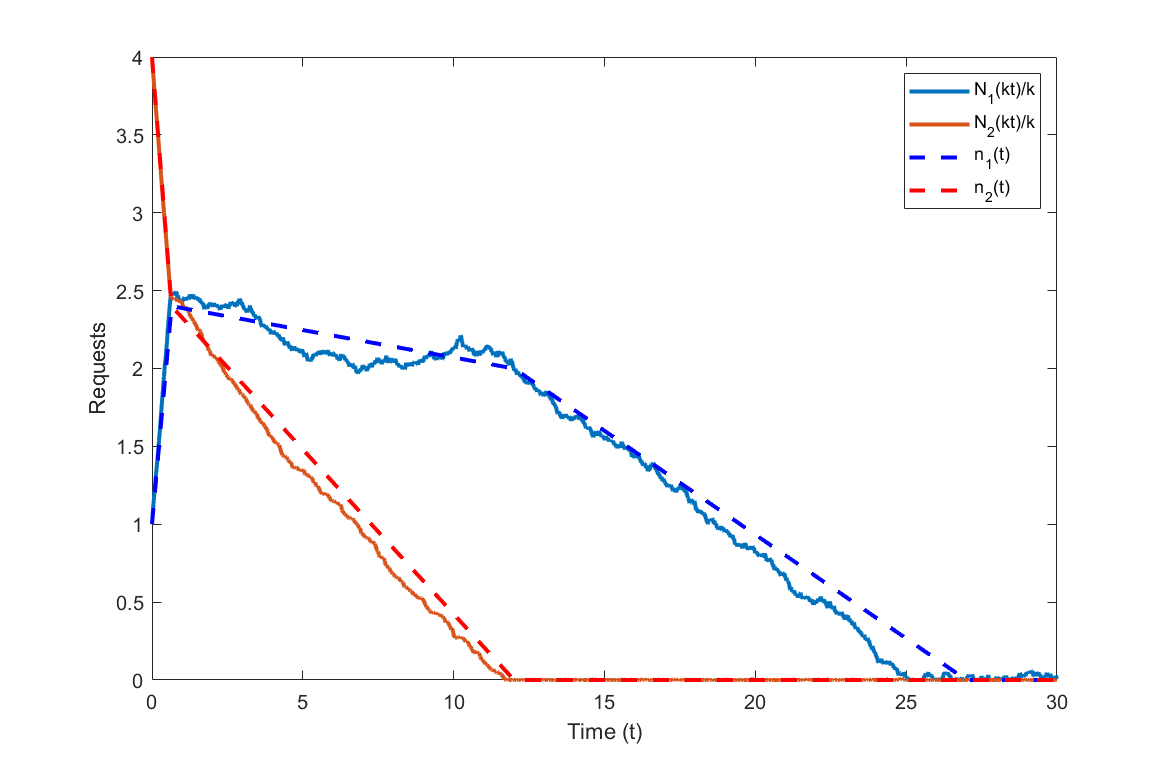}
    \caption{Case 2}
    \label{fig:case2}
    \end{subfigure}
    \caption{Sample trajectories of the number of requests in Case 1 and Case 2, both for the fluid limit and for the appropriately scaled stochastic system (with $k=1000$)}
\end{figure}

Note that the right-hand side in the two equations above are the minimum drifts that $n_1(\cdot)$ and $n_2(\cdot)$ can ever achieve while they are both positive, and the left hand side is the average of their drifts when they are both positive. Thus, when $n_1(t)=n_2(t)$, these conditions imply that both coordinates have the same drift, so
\begin{align*}
\quad \quad \frac{dn_1(t)}{dt}=\frac{dn_2(t)}{dt} = \frac{\lambda_1+\lambda_2-C_{1,2}}{2}. 
\end{align*}
While the largest coordinate will always decrease in this case, the smaller one can either increase or decrease. The key feature of this case is that if/when they intersect, they will remain together until they hit zero (as in Figure~\ref{fig:case1}). This is perhaps the typical behavior that one would expect of a pair of queues scheduled by a Max-Weight policy, which tends to equalize queues and exhibits state space collapse in heavy traffic, as seen in \cite{MW}.

\noindent{\bf Case 2:} Suppose that the assumptions of Theorem \ref{thm:uniquenessAndStability} hold, and that either for $i=1$ or for $i=2$, $\lambda_1 + \lambda_2 - C_{1,2} < 2\left(\lambda_i - \overline{C}_i\right) \text{ and } \lambda_i < \overline{C}_i.$
This corresponds to the case where one of the arrival rates is relatively larger than the other, but not extremely large (area shaded in turquoise in Figure \ref{fig:rateCases}). The first equation implies that both coordinates can never have the same drift while they are both positive. Therefore, if the coordinates become equal and positive at some point, they will not remain equal. Moreover, the second equation implies that the minimum drift while they are both positive is still negative, and thus the largest queue will always be decreasing. The trajectories of the processes $n_1(\cdot)$ and $n_2(\cdot)$ in this case look like the ones in Figure~\ref{fig:case2}. Since $n_1(0)>n_2(0)$, at the beginning $n_1(\cdot)$ will be decreasing and $n_2(\cdot)$ will be increasing. When they intersect, instead of remaining together as in Case 1, $n_2(\cdot)$ decreases at a slower rate than $n_1(\cdot)$. Finally, after $n_1(\cdot)$ hits zero, $n_2(\cdot)$ decreases at a higher rate until it hits zero as well.

\noindent{\bf Case 3:} Suppose that the assumptions of Theorem \ref{thm:uniquenessAndStability} hold, and that for $i=1$ or $i=2$, $\lambda_1 + \lambda_2 - C_{1,2} < 2\left(\lambda_i - \overline{C}_i\right) \text{ and } \lambda_i \geq \overline{C}_i.$

\begin{wrapfigure}{R}{0.5\textwidth}
    \centering
\includegraphics[scale=0.3]{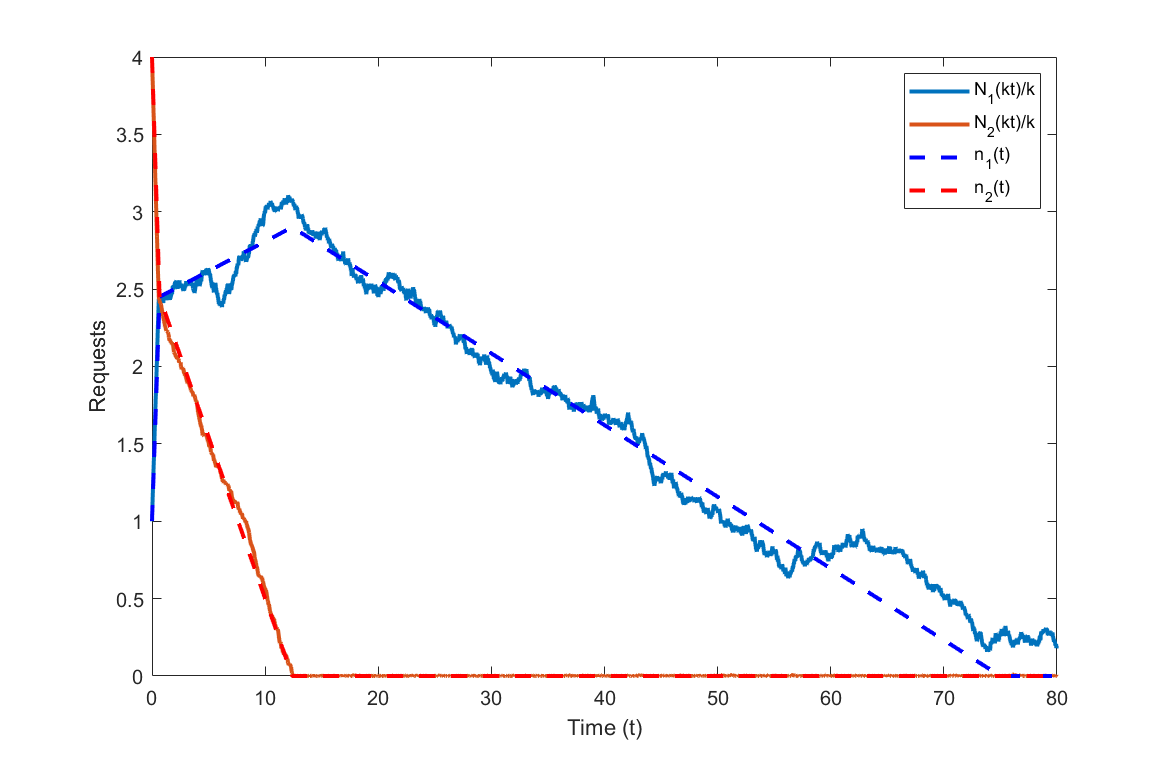}
    \caption{Sample trajectories of the number of requests in Case 3, both for the fluid limit and for the appropriately scaled stochastic system (with $k=1000$).}
    \label{fig:case3}
\end{wrapfigure}

This corresponds to the case where one of the arrival rates is almost as large as it can be while being stable, and therefore also relatively larger than the other arrival rate (area shaded in blue in Figure \ref{fig:rateCases}). These equations imply that both coordinates cannot have the same drift while they are both positive, and that the minimum drift for either $n_2(\cdot)$ or $n_1(\cdot)$ is non-negative while both coordinates are positive. Therefore, even if the coordinates become equal and positive at some point, they will not remain equal. Moreover, one of the coordinates is non-decreasing as long as the other is also positive, even when it is the largest one, and thus given the highest priority by our Max-Weight policy. For example. the trajectories of the processes $n_1(\cdot)$ and $n_2(\cdot)$ in this case could look like the ones in Figure~\ref{fig:case3}. Until they intersect, the trajectories behave the same as in cases 1 and 2. However, after they intersect, $n_2(\cdot)$ does not decrease, and in fact it increases if the inequality in the second equation is strict. This goes on until $n_1(\cdot)$ hits zero, at which point $n_1(\cdot)$ starts to decrease until it hits zero as well. This happens because, even when the first requests queue has strict priority, the non-idling condition forces the system into performing matchings involving the second request queue when there are not enough qubits to match with the first one, thus reducing the overall throughput of the latter. This is surprising because, when Max-Weight is used, the largest queue is usually decreasing instead of increasing as in this case.

\section{Conclusions and future work}\label{sec:conclusions}

In this paper we presented a general queueing model for a quantum switch. Due to its complexity, we only analyzed it for a couple of simple topologies. We showed that the Max-Weight policy is throughput optimal, and that the closure of its stability region is larger than the convex hull of the throughputs obtained when the requests queues are infinitely backlogged (which correspond to only cases 1 and 2 in Figure \ref{fig:rateCases}), a setting that was studied in the literature. Moreover, using a fluid limit approximation we showed that the decoherence of the qubits can produce the counter-intuitive phenomenon where the largest queue has positive drift under a Max-Weight policy, even when the overall system is stable.

There are several avenues of future work that we may pursue. The main one would be to extend our analysis to more complex networks. We conjecture that, in such networks, the stability region can still be obtained as the convex hull of all throughputs obtained by partially infinitely backlogging the system, in all possible ways. Therefore, this presents a combinatorial challenge that needs to be solved. The other directions would involve generalizing our model to include other features, such as the following.
\begin{itemize}
    \item[(i)] Heterogeneous abandonment and/or multi-way matching success probabilities.
    \item[(ii)] The possibility for more than two nodes to be simultaneously connected through a more complex entanglement process with more than two qubits.
    \item[(iii)] More general abandonment (decoherence) distributions, and the impact of different scheduling disciplines in that case.
\end{itemize}

\section*{Acknowledgements}
This work was partially supported by NSF Grant CMMI-2140534.

\bibliographystyle{imsart-nameyear}
\bibliography{references}



\begin{appendix}

\section{Proof of Lemma \ref{lem:basicStability}}\label{app:basicStability}

\begin{itemize}
    \item [(i)] First note that, when $\lambda<\mu$, a negative state is accessible from any positive state. Moreover, abandonments make sure that $0$ is accessible from any negative state. Therefore, the Markov chain $Q(\cdot)$ is irreducible. Since abandonments can also make the chain go from any negative state to $0$ in one step, $Q(\cdot)$ is also aperiodic.\\
    
    Consider the Lyapunov function $L(Q)=|Q|$. Then, using the triangular inequality and the law of total probability, we have
    \begin{align*}
        &\mathbb{E}[L(Q(t+1))-L(Q(t)) \mid Q(t)=q ] \\
        &\qquad\qquad\qquad = \mathbb{E}[ | Q(t) + A_1(t) - A_2(t) + D_2(t) | - |Q(t)| \mid Q(t)=q ] \\
        &\qquad\qquad\qquad \leq \mathbb{E}[ A_1(t) + | Q(t) - A_2(t) + D_2(t) | - |Q(t)| \mid Q(t)=q ] \\
        &\qquad\qquad\qquad = \lambda + \mathbb{E}[ | Q(t) - A_2(t) | - |Q(t)| \mid Q(t)=q ]\mathds{1}_{\{q\geq 0\}} \\
        &\qquad\qquad\qquad\qquad\qquad + \mathbb{E}[ | Q(t) - A_2(t) + D_2(t) | - |Q(t)| \mid Q(t)=q ]\mathds{1}_{\{q < 0\}} \\
        &\qquad\qquad\qquad \leq \lambda + \mathbb{E}[ | Q(t) - A_2(t) | - |Q(t)| \mid Q(t)=q ]\mathds{1}_{\{q\geq 0\}} \\
        &\qquad\qquad\qquad\qquad\qquad + \mathbb{E}[ A_2(t) + | Q(t) + D_2(t) | - |Q(t)| \mid Q(t)=q ]\mathds{1}_{\{q < 0\}} \\
        &\qquad\qquad\qquad = \lambda + \Big( \mathbb{E}[ - A_2(t) \mid Q(t)=q,\,\, A_2(t) \leq q ] \mathbb{P}(A_2(t) \leq q) \\
        &\qquad\qquad\qquad\qquad + \mathbb{E}[ A_2(t) - 2 Q(t) \mid Q(t)=q, \,\, A_2(t) > q ] \mathbb{P}(A_2(t) > q) \Big) \mathds{1}_{\{q\geq 0\}} \\
        &\qquad\qquad\qquad\qquad\qquad + \mathbb{E}[ A_2(t) - D_2(t) \mid Q(t)=q ]\mathds{1}_{\{q < 0\}} \\
        &\qquad\qquad\qquad \leq \lambda - \Big( \mathbb{E}[A_2(t) \mid A_2(t) \leq q ] \mathbb{P}\big(A_2(t) \leq q\big) - \mathbb{E}[A_2(t) \mid A_2(t) > q ] \mathbb{P}\big(A_2(t) > q\big) \Big)\mathds{1}_{\{q\geq 0\}} \\
        &\qquad\qquad\qquad\qquad\qquad\qquad\qquad\qquad\qquad\qquad\qquad\qquad\qquad\qquad\qquad\qquad\qquad + (\mu-\gamma|q|)\mathds{1}_{\{q< 0\}}.
    \end{align*}
    Since $\mathbb{E}[A_2(t)]=\mu>\lambda$, there exists $\bar{A}$ large enough such that
    \[ \mathbb{E}\left[A_2(t) \mid A_2(t) \leq \bar{A} \right] \mathbb{P}\left(A_2(t) \leq \bar{A}\right) - \mathbb{E}\left[A_2(t) \mid A_2(t) > \bar{A} \right] \mathbb{P}\left(A_2(t) > \bar{A}\right) \geq \frac{\mu+\lambda}{2}. \]
    If $|q|\geq\max\left\{\bar{A}, 2\mu/\gamma\right\}$, we have
    \begin{align*}
        \mathbb{E}\big[L(Q(t+1))-L(Q(t)) \mid Q(t)=q \big] &\leq -\left( \frac{\mu-\lambda}{2} \right) < 0.
    \end{align*}
    Combining this with the facts that there are only finitely many states with $|q|<\max\left\{\bar{A}, 2\mu/\gamma\right\}$, and that $Q(\cdot)$ is irreducible and aperiodic, we conclude that $Q(\cdot)$ is positive recurrent by the Foster-Lyapunov theorem \cite{meyn_tweedie_1992}.
    
    Finally, since $Q(\cdot)$ is positive recurrent, and since the throughput is equal to the time average departure rate of the queue with no abandonments, then the throughput is also equal to its arrival rate  $\lambda$.
    \item [(ii)] Note that 
    \begin{align*}
        Q(t) &= Q(0) + \sum\limits_{k=0}^{t-1} \big[ A_1(k) - A_2(k) + D_2(k) \big] \\
        &\geq Q(0) + \sum\limits_{k=0}^{t-1} \big[ A_1(k) - A_2(k) \big].
    \end{align*}
    Therefore, $Q(\cdot)$ is lower bounded by a random walk. Since $\lambda>\mu$, this is a positively biased random walk that diverges to $+\infty$ a.s., and thus so does $Q(\cdot)$. In particular, this implies that for almost every sample path $\omega$, there exists a time $T(\omega)$ such that $Q(\omega,t)>0$ for all $t\geq T(\omega)$. It follows that $M(\omega,t)=A_2(\omega,t)$ for all $t\geq T(\omega)$, and thus
    \begin{align*}
        \text{Throughput} &= \liminf\limits_{t\to\infty} \frac{1}{t} \sum\limits_{k=1}^t M(\omega,k) \\
        &= \liminf\limits_{t\to\infty} \frac{1}{t} \sum\limits_{k=1}^{T(\omega)} M(\omega,k) + \frac{1}{t} \sum\limits_{k=T(\omega)+1}^t A_2(\omega,k),
    \end{align*}
    which converges to $\mu_2$ for almost every sample path $\omega$.
    \end{itemize}

\section{Proof of Lemma \ref{lem:trivialStability}} \label{app:trivialStability}
The irreducibility and aperiodicity of this chain follows from the same argument as in the proof of Lemma \ref{lem:basicStability}. Consider the Lyapunov function $L(Q)=|Q|$. Then, using the triangular inequality, we get
    \begin{align*}
        \mathbb{E}[L(Q(t+1))-L(Q(t)) \mid Q(t)=q ] &\leq (\mu_1+\mu_2-\gamma_1|q|)\mathds{1}_{\{q\geq 0\}} + (\mu_1+\mu_2-\gamma_2|q|)\mathds{1}_{\{q< 0\}}.
    \end{align*}
    If
    \begin{equation} 
        |q|\geq \frac{\mu_1+\mu_2 + \epsilon}{\min\{\gamma_1,\gamma_2\}}, \label{eq:fewStates}
    \end{equation}
    for some $\epsilon>0$, we have
    \begin{align*}
        \mathbb{E}\big[L(Q(t+1))-L(Q(t)) \,\big|\, Q(t)=q \big] &\leq -\epsilon < 0.
    \end{align*}
    Combining this with the fact that there are finitely many states where Equation \eqref{eq:fewStates} does not hold, and with the irreducibility and aperiodicity of $Q(\cdot)$, it follows that $Q(\cdot)$ is positive recurrent with unique invariant distribution $\pi$ by the Foster-Lyapunov theorem \cite{meyn_tweedie_1992}. Finally, the recursive definition of $Q(\cdot)$ implies that
    \[ M(t) = Q^+(t) - Q^+(t+1) + A_1(t) - D_1(t). \]
    Moreover, when $Q(\cdot)$ starts in steady-state, we have $\mathbb{E}_\pi[Q^+(t)]=\mathbb{E}_\pi[Q^+(t+1)]<\infty$. Therefore, in steady-state, we have
    \begin{align*}
        \mathbb{E}_\pi[M(t)] &= \mathbb{E}_\pi[Q^+(t) - Q^+(t+1) + A_1(t) - D_1(t)] \\
        &= \mathbb{E}_\pi[A_1(t) - D_1(t)] \\
        &= \mu_1 - \gamma_1 \mathbb{E}_\pi[Q^+].
    \end{align*}
    It follows that
    \begin{align*}
        \text{Throughput} &= \liminf\limits_{t\to\infty} \frac{1}{t} \sum\limits_{k=1}^t M(k) \\
        &= \mathbb{E}_\pi[M(1)] \\
        &= \mu_1 - \gamma_1 \mathbb{E}_\pi[Q^+].
    \end{align*}
    The argument for $Q^-$ is analogous.

\section{Proof of Theorem \ref{thm:initialStability}} \label{app:initialStability}

We present the proof of each part of Theorem \ref{thm:initialStability} separately.

\subsection{Positive recurrence}
The positive recurrence is established in two steps. First, we construct a pair of coupled stochastic processes, and show that these new processes stochastically dominate the original one. Second, we establish the positive recurrence of these new processes through a multi-step Lyapunov drift argument. Combining the two steps, the positive recurrence of the original process follows. The two processes used for coupling are defined below. 

\textbf{Process I:}  We first consider a coupled process $\left(\overline N(\cdot), \underline{\bf  Q}(\cdot)\right)$, where qubits are always matched regardless of whether there are requests or not. In particular, this means that the qubit queues behave as if the requests were infinitely backlogged, regardless of the number of requests in their queue. Thus, the qubits' queues act as a two-way matching system. The coupling with the original processes is done in a way so that these new ones have the same initial condition, i.e., with $\left(\overline N(0), \underline{\bf  Q}(0)\right)=\big( N(0), {\bf Q}(0) \big)$, and so that they have the same arrival processes $A(\cdot)$, $S_1(\cdot)$ and $S_2(\cdot)$, and the same abandonment primitives $Z_\ell^{(i)}(\cdot)$. These new processes are defined recursively as
\begin{align*}
    \overline N(t+1) &= \overline N(t) + A(t) - \sum\limits_{\ell=1}^{\overline M(t)} Y_\ell(t) \\
    \underline Q_1(t+1) &= \underline Q_1(t) - \underline D_1(t) + S_1(t) - \overline{\overline{M}}(t) \\
    \underline Q_2(t+1) &= \underline Q_2(t) - \underline D_2(t) +S_2(t) - \overline{\overline{M}}(t),
\end{align*}
where
\[ \overline{\overline{M}}(t) = \min\left\{\underline Q_1(t)-\underline D_1(t)+S_1(t),\,\, \underline Q_2(t)-\underline D_2(t)+S_2(t)\right\}, \]
is the total number of two-way and three-way matchings made in the $t$-th slot (including the ones that do not involve a request) and
\[ \overline M(t) = \min\left\{ \overline{\overline{M}}(t) , \min\left\{ m :\sum\limits_{\ell=1}^m Y_\ell(t) \geq \overline N(t)+A(t) \right\}\right\} \]
is the number of successful three-way matchings in the $t$-th slot. Also, we couple the abandonment of this process with the original one using
\[ \underline{D}_i(t) = \sum\limits_{\ell=1}^{\underline{Q}_i(t)} Z^{(i)}_\ell(t). \]

Note that this new system effectively discards qubits that could have been matched to a request in the future if there are no pending requests in the present. This stands in contrast with the original system, where the qubits are stored until a new request arrives. Intuitively, there should be less qubits in the queues for this new system, and the number of queued requests should higher in this new system. This intuition is formalized in the following two lemmas.

\begin{lemma}\label{lem:Qdominance}
We have $Q_i(t) \geq \underline Q_i(t)$, for all $t\geq 0$ and $i\in\{1,2\}$.
\end{lemma}
\begin{proof}{Proof:}
Note that, for each sample path, $\underline Q_i(\cdot)$ can be obtained from $Q_i(\cdot)$ by simply shortening the ``service time'' of qubits so that they depart earlier. Therefore, \cite[Proposition 2]{monotonicity} implies that $Q_i(t) \geq \underline Q_i(t)$, for all $t\geq 0$.
  \end{proof}

Suppose that $p=1$, that is, suppose that three-way matchings succeed with probability one. Then, it can be checked that, for each sample path,
$\overline N(\cdot)$ can be obtained from $N(\cdot)$ by simply adding $\overline{\overline{M}}(\cdot)-\overline{M}(\cdot)$ request arrivals so that they match with the extra qubits that were matched in the $\overline{N}(\cdot)$ system and not in the $N(\cdot)$ system. In this case, \cite[Proposition 2]{monotonicity} implies that $\overline N(t) \geq N(t)$, for all $t\geq 0$. However, when $p<1$ this is no longer the case, as added requests could remain in the system after exhausting all extra qubits. As a result, we prove this stochastic dominance in the following lemma.

\begin{lemma}\label{lem:Ndominance}
We have $\overline N(t) \geq N(t)$, for all $t\geq 0$.
\end{lemma}
\begin{proof}{Proof:}
We prove this using induction. At $t =0$, the condition holds by the fact that we have $\left(\overline N(0), \underline{\bf  Q}(0)\right)=\big( N(0), {\bf Q}(0) \big)$ by construction. Now, assume 
$\overline N(t) \geq N(t)$. Then,
\begin{align}
    \overline N(t+1) - N(t+1) \nonumber & = \overline N(t) - \sum\limits_{\ell=1}^{\overline M(t)} Y_\ell(t) - \left[N(t) - \sum\limits_{\ell=1}^{M(t)} Y_\ell(t)  \right] \nonumber \\
    & \geq \sum\limits_{\ell=1}^{M(t)} Y_\ell(t) - \sum\limits_{\ell=1}^{\overline M(t)} Y_\ell(t) \nonumber \\
    & \geq \sum\limits_{\ell=1}^{M(t)} Y_\ell(t) - \sum\limits_{\ell=1}^{\overline{\overline{M}}(t)} Y_\ell(t), \label{eq:lower_bound} 
\end{align}
where the last inequality follows by using $\overline{M}(t) \leq \overline{\overline{M}}(t)$ by construction. Moreover, we have
\begin{align}
    M(t) &= \min\left\{Q_1(t)-D_1(t)+S_1(t),\,\, Q_2(t)-D_2(t)+S_2(t),\,\, U(t) \right\} \nonumber \\
    &\leq \min\left\{Q_1(t)-D_1(t)+S_1(t),\,\, Q_2(t)-D_2(t)+S_2(t) \right\} \nonumber \\
    &= \min\left\{\underline Q_1(t)-\underline D_1(t)+S_1(t),\,\, \underline Q_2(t)-\underline D_2(t)+S_2(t) \right\} \nonumber \\
    & + \left[\min\left\{Q_1(t)-D_1(t)+S_1(t),\,\, Q_2(t)-D_2(t)+S_2(t) \right\} - \underline Q_1(t) + \underline D_1(t)  \right]\mathds{1}_{\left\{\underline{Q}_1(t)- \underline{D}_1(t)+S_1(t) < \underline{Q}_2(t)- \underline{D}_2(t)+S_2(t)\right\}} \nonumber \\
    &+ \left[\min\left\{Q_1(t)-D_1(t)+S_1(t),\,\, Q_2(t)-D_2(t)+S_2(t) \right\} - \underline Q_2(t) + \underline D_2(t) \right]\mathds{1}_{\left\{ \underline{Q}_1(t)- \underline{D}_1(t)+S_1(t) \geq \underline{Q}_2(t)- \underline{D}_2(t) +S_2(t)\right\}} \nonumber \\
    &\leq \min\left\{\underline Q_1(t)-\underline D_1(t)+S_1(t),\,\, \underline Q_2(t)-\underline D_2(t)+S_2(t) \right\} \nonumber \\
    &\qquad\qquad\qquad + \left[Q_1(t) -D_1(t) - \underline Q_1(t) + \underline D_1(t)  \right]\mathds{1}_{\left\{\underline{Q}_1(t)- \underline{D}_1(t)+S_1(t) < \underline{Q}_2(t)- \underline{D}_2(t)+S_2(t)\right\}} \nonumber \\
    &\qquad\qquad\qquad\qquad\qquad+ \left[Q_2(t) -D_2(t) - \underline Q_2(t) + \underline D_2(t) \right]\mathds{1}_{\left\{ \underline{Q}_1(t)- \underline{D}_1(t)+S_1(t) \geq \underline{Q}_2(t)- \underline{D}_2(t) +S_2(t)\right\}} \nonumber \\
    &= \overline{\overline{M}}(t) + \left[Q_1(t) -D_1(t) - \underline Q_1(t) + \underline D_1(t)  \right]\mathds{1}_{\left\{\underline{Q}_1(t)- \underline{D}_1(t)+S_1(t) < \underline{Q}_2(t)- \underline{D}_2(t)+S_2(t)\right\}} \nonumber \\
    &\qquad\qquad\qquad+ \left[Q_2(t) -D_2(t) - \underline Q_2(t) + \underline D_2(t) \right]\mathds{1}_{\left\{ \underline{Q}_1(t)- \underline{D}_1(t)+S_1(t) \geq \underline{Q}_2(t)- \underline{D}_2(t) +S_2(t)\right\}}, \label{eq:split}
\end{align}
Note that Lemma~\ref{lem:Qdominance} implies that $Q_i(t) \geq \underline Q_i(t)$, and thus
\begin{equation*} \label{eq:q-d}
Q_i(t) -D_i(t) - \underline Q_i(t) + \underline D_i(t) = \sum\limits_{\ell=1}^{Q_i(t)} 1-Z^{(i)}_\ell(t) - \sum\limits_{\ell=1}^{\underline Q_i(t)} 1-Z^{(i)}_\ell(t) \geq 0.
\end{equation*}
Combining this with Equation~\eqref{eq:split}, it follows that $M(t)\geq \overline{\overline{M}}(t)$. Substituting this in Equation~\eqref{eq:lower_bound} concludes the induction step, and the proof.

\end{proof}

\textbf{Process II:} We now consider a pair of coupled processes $\overline Q_1(\cdot)$ and $\overline Q_2(\cdot)$, where qubits are never matched and so they only depart when they abandon. In particular, these processes are GI/M/$\infty$ queues with potentially correlated services. Fortunately, the potential service correlations do not affect the drift of the queues. The coupling with the original processes is done in a way so that these new ones have the same initial conditions, i.e., with $\overline Q_1(0)=Q_1(0)$ and $\overline Q_2(0)=Q_2(0)$, and so that they have the same arrival processes $S_1(\cdot)$ and $S_2(\cdot)$, and the same abandonment primitives $Z_\ell^{(i)}(\cdot)$. These new processes are defined recursively as
\begin{align*}
    \overline Q_1(t+1) &= \overline Q_1(t) - \overline D_1(t) + S_1(t) \\
    \overline Q_2(t+1) &= \overline Q_2(t) - \overline D_2(t) +S_2(t),
\end{align*}
where
\[ \overline D_i(t) = \sum\limits_{\ell=1}^{\overline Q_i(t)} Z^{(i)}_\ell(t). \]

\begin{lemma}\label{lem:infinityQueuesBound}
We have $\overline Q_i(t) \geq Q_i(t)$ and $\overline Q_i(t) \geq \underline Q_i(t)$, for all $t\geq 0$, for all $i$.
\end{lemma}
\begin{proof}{Proof:}
Suppose that the inequalities hold at time $t$. Then
\begin{align*}
    \overline Q_i(t+1) - Q_i(t+1) &= \left[\overline Q_i(t) - \overline D_i(t)\right] - \big[ Q_i(t) - D_i(t) \big] + M(t) \\
    &\geq \left[\overline Q_i(t) - Q_i(t) \right] +  \left[ \sum\limits_{\ell=1}^{Q_i(t)} Z^{(i)}_\ell(t) - \sum\limits_{\ell=1}^{\overline Q_i(t)} Z^{(i)}_\ell(t) \right] \\
    &\geq 0.
\end{align*}
Moreover,
\begin{align*}
    \overline Q_i(t+1) - \underline Q_i(t+1) &= \left[\overline Q_i(t) - \overline D_i(t)\right] - \left[ \underline Q_i(t) - \underline D_i(t) \right] + \overline{\overline{M}}(t) \\
    &\geq \left[\overline Q_i(t) - \underline Q_i(t) \right] + \left[ \sum\limits_{\ell=1}^{\underline Q_i(t)} Z^{(i)}_\ell(t) - \sum\limits_{\ell=1}^{\overline Q_i(t)} Z^{(i)}_\ell(t) \right] \\
    &\geq 0.
\end{align*}
\end{proof}

\textbf{Lyapunov-Drift:} It follows from the previous two lemmas that the positive recurrence of $\big(\overline N(\cdot), \underline{\bf  Q}(\cdot), \overline{\bf  Q}(\cdot)\big)$ implies the positive recurrence of $\big(N(\cdot), {\bf Q}(\cdot)\big)$. To establish the positive recurrence of this new process, we define the Lyapunov function
\[ L\left(\overline n, \underline{\bf q}, \overline{\bf q} \right) = \overline n + \overline q_1 + \overline q_2. \]

\begin{lemma}
  There exists constants $K,T,\delta>0$ such that, if
  \[ L\left(\overline N(t), \underline{\bf Q}(t), \overline{\bf Q}(t) \right) > 3KT. \]
  then
  \[ \mathbb{E}\left[ L\left(\overline N(t+T), \underline{\bf Q}(t+T), \overline{\bf Q}(t+T) \right)- L\left(\overline N(t), \underline{\bf Q}(t), \overline{\bf Q}(t) \right) \,\Big|\, \overline N(t), \underline{\bf Q}(t), \overline{\bf Q}(t) \right] \leq -\delta T. \]
\end{lemma}
\begin{proof}{Proof:}
Since $\underline{\bf Q}(\cdot)$ behaves path-wise the same as a two-sided queue with abandonments (which is ergodic and has throughput $C_Y$ in steady-state), for every $\epsilon>0$ there exists a constant $T_\epsilon$ large enough such that
\begin{align}
    \mathbb{E}\left[ \left. \sum\limits_{k=0}^{T_\epsilon-1} \sum\limits_{\ell=1}^{\overline{\overline{M}}(t+k)} Y_\ell(t+k) \,\right|\, \overline N(t), \underline{\bf Q}(t), \overline{\bf Q}(t) \right] &= p\mathbb{E}\left[ \left. \sum\limits_{k=0}^{T_\epsilon-1} \overline{\overline{M}}(t+k) \,\right|\, \left(\underline Q_1(t), \underline Q_2(t)\right) \right] \nonumber \\
    &\geq T_\epsilon \Big( p C_Y - \epsilon \Big). \label{eq:lowerBound}
\end{align} 
Moreover, we have
\begin{align*}
    \sum\limits_{k=0}^{T-1} \sum\limits_{\ell=1}^{\overline{\overline{M}}(t+k)} Y_\ell(t+k) &\leq \sum\limits_{k=0}^{T-1} \overline{\overline{M}}(t+k) \\
    &= \sum\limits_{k=0}^{T-1} \min\left\{\underline Q_1(t+k)-\underline D_1(t+k)+S_1(t+k),\,\, \underline Q_2(t+k)-\underline D_2(t+k)+S_2(t+k)\right\},
\end{align*}
where the right-hand side is the total number of matchings made between the qubits in our coupled process, from time $t$ to time $t+T-1$. Therefore, this must be less than or equal to the number of qubits at time $t$ plus all the qubits that arrived during that period of time. That is, we have
\begin{align*}
    \sum\limits_{k=0}^{T-1} \sum\limits_{\ell=1}^{\overline{\overline{M}}(t+k)} Y_\ell(t+k) &\leq \underline Q_1(t) + \underline Q_2(t) + \sum\limits_{k=0}^{T-1} S_1(t+k) + S_2(t+k) \\
    &\leq \overline Q_1(t) + \overline Q_2(t) + \sum\limits_{k=0}^{T-1} S_1(t+k) + S_2(t+k),
\end{align*}
where the last inequality is due to Lemma \ref{lem:infinityQueuesBound}. In particular, if $\overline Q_1(t) + \overline Q_2(t)\leq KT$, we have
\begin{align*}
    \frac{1}{T} \sum\limits_{k=0}^{T-1} \sum\limits_{\ell=1}^{\overline{\overline{M}}(t+k)} Y_\ell(t+k) &\leq K+ \frac{1}{T} \sum\limits_{k=0}^{T-1} S_1(t+k) + S_2(t+k).
\end{align*}
Since the right-hand converges in probability to $K+\mu_1+\mu_2$ as $T\to\infty$, if $\overline Q_1(t) + \overline Q_2(t)\leq KT$ we have that
\begin{align*}
    &\mathbb{E}\left[ \left. \min\left\{ \frac{1}{T} \sum\limits_{k=0}^{T-1} \sum\limits_{\ell=1}^{\overline{\overline{M}}(t+k)} Y_\ell(t+k), \,\, 2K \right\} \,\right|\, \overline N(t), \underline{\bf Q}(t), \overline{\bf Q}(t) \right] \nonumber \\
    &\qquad\qquad\qquad\qquad\qquad\qquad\qquad\qquad - \mathbb{E}\left[ \left. \frac{1}{T} \sum\limits_{k=0}^{T-1} \sum\limits_{\ell=1}^{\overline{\overline{M}}(t+k)} Y_\ell(t+k) \,\right|\, \overline N(t), \underline{\bf Q}(t), \overline{\bf Q}(t) \right] \nonumber \\
    &= 2K \mathbb{P}\left( \left. \frac{1}{T} \sum\limits_{k=0}^{T-1} \sum\limits_{\ell=1}^{\overline{\overline{M}}(t+k)} Y_\ell(t+k) > 2K \,\right|\, \overline N(t), \underline{\bf Q}(t), \overline{\bf Q}(t) \right)
\end{align*}
converges to zero as $T\to\infty$ for all $K>\mu_1+\mu_2$. Therefore, for every $\epsilon>0$ there exist constants $K$ and $T_\epsilon$ large enough such that, if $\overline Q_1(t) + \overline Q_2(t)\leq KT_\epsilon$, we have
\begin{align}
    &\mathbb{E}\left[ \left. \min\left\{ \sum\limits_{k=0}^{T_\epsilon-1} \sum\limits_{\ell=1}^{\overline{\overline{M}}(t+k)} Y_\ell(t+k), \,\, 2K T_\epsilon \right\} \,\right|\, \overline N(t), \underline{\bf Q}(t), \overline{\bf Q}(t) \right] \nonumber \\
    &\qquad\qquad\qquad\qquad\qquad\qquad\qquad\qquad \geq \mathbb{E}\left[ \left. \sum\limits_{k=0}^{T_\epsilon-1} \sum\limits_{\ell=1}^{\overline{\overline{M}}(t+k)} Y_\ell(t+k) \,\right|\, \overline N(t), \underline{\bf Q}(t), \overline{\bf Q}(t) \right] - T_\epsilon\epsilon. \label{eq:upperBound}
\end{align}
It follows that, if $\overline Q_1(t) + \overline Q_2(t)\leq K T_\epsilon$ and $\overline N(t)\geq 2K T_\epsilon$, for all $\epsilon$ small enough, we have 
\begin{align}
    &\mathbb{E} \left[\overline N(t+T_\epsilon)- \overline N(t) \,\Big|\, \overline N(t), \underline{\bf Q}(t), \overline{\bf Q}(t) \right] \nonumber \\
    &\qquad\qquad\qquad\qquad\qquad = \sum\limits_{k=0}^{T_\epsilon-1} \mathbb{E}\left[ A(t+k) - \sum\limits_{\ell=1}^{\overline{M}(t+k)} Y_\ell(t+k) \,\Big|\, \overline N(t), \underline{\bf Q}(t), \overline{\bf Q}(t) \right] \nonumber \\
    &\qquad\qquad\qquad\qquad\qquad = \lambda T_\epsilon - \mathbb{E}\left[ \left. \sum\limits_{k=0}^{T_\epsilon-1} \sum\limits_{\ell=1}^{\overline{M}(t+k)} Y_\ell(t+k) \,\right|\, \overline N(t), \underline{\bf Q}(t), \overline{\bf Q}(t) \right] \nonumber \\
    &\qquad\qquad\qquad\qquad\qquad \leq \lambda T_\epsilon - \mathbb{E}\left[ \left. \min\left\{ \sum\limits_{k=0}^{T_\epsilon-1} \sum\limits_{\ell=1}^{\overline{M}(t+k)} Y_\ell(t+k), \overline N(t) \right\} \,\right|\, \overline N(t), \underline{\bf Q}(t), \overline{\bf Q}(t) \right]. \label{eq:interm}
\end{align}
Note that, when
\[ \sum\limits_{k=0}^{T_\epsilon-1} \sum\limits_{\ell=1}^{\overline{M}(t+k)} Y_\ell(t+k) \geq \overline N(t), \]
then we also have
\[ \sum\limits_{k=0}^{T_\epsilon-1} \sum\limits_{\ell=1}^{\overline{\overline{M}}(t+k)} Y_\ell(t+k) \geq \overline N(t), \]
and thus
\[ \min\left\{ \sum\limits_{k=0}^{T_\epsilon-1} \sum\limits_{\ell=1}^{\overline{\overline{M}}(t+k)} Y_\ell(t+k), \overline N(t) \right\} = \min\left\{ \sum\limits_{k=0}^{T_\epsilon-1} \sum\limits_{\ell=1}^{\overline{M}(t+k)} Y_\ell(t+k), \overline N(t) \right\}. \]
On the other hand, when
\[ \sum\limits_{k=0}^{T_\epsilon-1} \sum\limits_{\ell=1}^{\overline{M}(t+k)} Y_\ell(t+k) <\overline N(t), \]
we have that
\[ \overline{N}(t+k) = \overline{N}(t) + \sum\limits_{i=0}^{k-1} A(t+i) - \sum\limits_{i=0}^{k-1} \sum\limits_{\ell=1}^{\overline{M}(t+i)} Y_\ell(t+i) \geq \overline{N}(t) - \sum\limits_{i=0}^{k-1} \sum\limits_{\ell=1}^{\overline{M}(t+i)} Y_\ell(t+i) > 0 \]
for all $k\leq T_\epsilon$. Combining this with the fact that
\[ \overline M(t+i) = \min\left\{ \overline{\overline{M}}(t+i) , \min\left\{ m :\sum\limits_{\ell=1}^m Y_\ell(t+i) \geq \overline N(t+i)+A(t+i) \right\}\right\}, \]
it follows that $N(t+i)>0$ for all $i\leq T_\epsilon$ implies that $\overline{M}(t+i-1)=\overline{\overline{M}}(t+i-1)$ for all $k\leq T_\epsilon$. Therefore, we have
\[ \min\left\{ \sum\limits_{k=0}^{T_\epsilon-1} \sum\limits_{\ell=1}^{\overline{\overline{M}}(t+k)} Y_\ell(t+k), \overline N(t) \right\} = \min\left\{ \sum\limits_{k=0}^{T_\epsilon-1} \sum\limits_{\ell=1}^{\overline{M}(t+k)} Y_\ell(t+k), \overline N(t) \right\}. \]
Combining this with Equation \eqref{eq:interm} we get that
\begin{align}
    &\mathbb{E} \left[\overline N(t+T_\epsilon)- \overline N(t) \,\Big|\, \overline N(t), \underline{\bf Q}(t), \overline{\bf Q}(t) \right] \nonumber\\
    &\qquad\qquad\qquad\qquad\qquad\qquad\qquad \leq \lambda T_\epsilon - \mathbb{E}\left[ \left. \min\left\{ \sum\limits_{k=0}^{T_\epsilon-1} \sum\limits_{\ell=1}^{\overline{\overline{M}}(t+k)} Y_\ell(t+k), \overline N(t) \right\} \,\right|\, \overline N(t), \underline{\bf Q}(t), \overline{\bf Q}(t) \right] \nonumber \\
    &\qquad\qquad\qquad\qquad\qquad\qquad\qquad \leq \lambda T_\epsilon - \mathbb{E}\left[ \left. \min\left\{ \sum\limits_{k=0}^{T_\epsilon-1} \sum\limits_{\ell=1}^{\overline{\overline{M}}(t+k)} Y_\ell(t+k), 2K T_\epsilon \right\} \,\right|\, \overline N(t), \underline{\bf Q}(t), \overline{\bf Q}(t) \right] \nonumber \\
    &\qquad\qquad\qquad\qquad\qquad\qquad\qquad \overset{(a)}{\leq} \lambda T_\epsilon - \mathbb{E}\left[ \left. \sum\limits_{k=0}^{T_\epsilon-1} \sum\limits_{\ell=1}^{\overline{\overline{M}}(t+k)} Y_\ell(t+k) \,\right|\, \overline N(t), \underline{\bf Q}(t), \overline{\bf Q}(t) \right] + T_\epsilon\epsilon \nonumber \\
    &\qquad\qquad\qquad\qquad\qquad\qquad\qquad \overset{(b)}{\leq} \lambda T_\epsilon - T_\epsilon(p C_Y-\epsilon) + T_\epsilon \epsilon \nonumber \\
    &\qquad\qquad\qquad\qquad\qquad\qquad\qquad = -(p C_Y-2\epsilon - \lambda)T_\epsilon, \label{eq:drift1}
\end{align}
where (a) is due to Equation \eqref{eq:upperBound}, and (b) is due to Equation \eqref{eq:lowerBound}. Moreover, we always have
\begin{align}
    \mathbb{E}\left[\overline N(t+T_\epsilon)- \overline N(t) \,\Big|\, \overline N(t), \underline{\bf Q}(t), \overline{\bf Q}(t) \right] &= \sum\limits_{k=0}^{T_\epsilon-1} \mathbb{E}\left[ A(t+k) - \sum\limits_{\ell=1}^{\overline{M}(t+k)} Y_\ell(t+k) \,\Big|\, \overline N(t), \underline{\bf Q}(t), \overline{\bf Q}(t) \right] \nonumber \\
    &\leq \sum\limits_{k=0}^{T_\epsilon-1} \mathbb{E}\left[ A(t+k) \,\Big|\, \overline N(t), \underline{\bf Q}(t), \overline{\bf Q}(t) \right] \nonumber \\
    &= \lambda T_\epsilon. \label{eq:drift2}
\end{align}

On the other hand, recall that $\overline{Q}_1(\cdot)$ and $\overline{Q}_2(\cdot)$ are GI/M/$\infty$ queues with potentially correlated services. However, it is easily checked that these correlations do not affect the drift, and thus the expectation, of the queues. Therefore, we have
\begin{align}
    &\mathbb{E}\Big[\overline Q_1(t+T) + \overline Q_2(t+T) - \overline Q_1(t) - \overline Q_2(t) \,\Big|\, \overline N(t), \underline{\bf Q}(t), \overline{\bf Q}(t) \Big] \nonumber \\
    &\qquad\qquad\qquad\qquad\qquad\qquad = \left(\overline Q_1(t) - \frac{\mu_1}{\gamma_1}\right) \left[(1-\gamma_1)^T-1\right] + \left(\overline Q_2(t) - \frac{\mu_2}{\gamma_2}\right)\left[(1-\gamma_2)^T-1\right] \nonumber \\
    &\qquad\qquad\qquad\qquad\qquad\qquad \leq \frac{\mu_1}{\gamma_1} + \frac{\mu_2}{\gamma_2} - \overline Q_1(t)\left[1- (1-\gamma_1)^T\right] - \overline Q_2(t) \left[1- (1-\gamma_2)^T\right] \nonumber \\
    &\qquad\qquad\qquad\qquad\qquad\qquad \leq \frac{\mu_1}{\gamma_1} + \frac{\mu_2}{\gamma_2} - \frac{\overline Q_1(t)+ \overline Q_2(t)}{2}, \label{eq:drift3}
\end{align}
for all $T$ large enough.

Finally, let $\epsilon$ be small enough, and $T$ and $K$ be large enough so that equations \eqref{eq:drift1}, \eqref{eq:drift2}, and \eqref{eq:drift3} hold. If $\overline Q_1(t) + \overline Q_2(t)\leq K T$ and $\overline N(t)\geq 2K T$, then equations \eqref{eq:drift1} and \eqref{eq:drift3} imply
\begin{align*}
    &\mathbb{E}\left[ L\left(\overline N(t+T), \underline{\bf Q}(t+T), \overline{\bf Q}(t+T) \right)- L\left(\overline N(t), \underline{\bf Q}(t), \overline{\bf Q}(t) \right) \,\Big|\, \overline N(t), \underline{\bf Q}(t), \overline{\bf Q}(t) \right] \\
    &\qquad\qquad\qquad\qquad\qquad\qquad\qquad\qquad\qquad\qquad\qquad \leq -(p C_Y-2\epsilon - \lambda) T + \frac{\mu_1}{\gamma_1} + \frac{\mu_2}{\gamma_2} - \frac{\overline Q_1(t)+ \overline Q_2(t)}{2} \\
    &\qquad\qquad\qquad\qquad\qquad\qquad\qquad\qquad\qquad\qquad\qquad \leq -(p C_Y-2\epsilon - \lambda) T + \frac{\mu_1}{\gamma_1} + \frac{\mu_2}{\gamma_2},
\end{align*} 
which is negative for all $T$ large enough and $\epsilon$ small enough. Furthermore, if $\overline Q_1(t) + \overline Q_2(t) \geq K T$, equations \eqref{eq:drift2} and \eqref{eq:drift3} imply
\begin{align*}
    &\mathbb{E}\left[ L\left(\overline N(t+T), \underline{\bf Q}(t+T), \overline{\bf Q}(t+T) \right)- L\left(\overline N(t), \underline{\bf Q}(t), \overline{\bf Q}(t) \right) \,\Big|\, \overline N(t), \underline{\bf Q}(t), \overline{\bf Q}(t) \right] \\
    &\qquad\qquad\qquad\qquad\qquad\qquad\qquad\qquad\qquad\qquad\qquad\qquad\qquad\qquad \leq \lambda T + \frac{\mu_1}{\gamma_1} + \frac{\mu_2}{\gamma_2} - \frac{\overline Q_1(t)+ \overline Q_2(t)}{2} \\
    &\qquad\qquad\qquad\qquad\qquad\qquad\qquad\qquad\qquad\qquad\qquad\qquad\qquad\qquad \leq \lambda T + \frac{\mu_1}{\gamma_1} + \frac{\mu_2}{\gamma_2} - \frac{KT}{2},
\end{align*}
which is also negative for all $T$ large enough.
\end{proof}

Note that, since $\overline Q_i(t) \geq \underline Q_i(t)$ for all $t\geq 0$ (Lemma \ref{lem:infinityQueuesBound}), there are finitely many states such that $\overline n + \overline q_1 + \overline q_2\leq 3KT$. Moreover, the $T$ time-step drift of $L$ is negative for all $T$ and $K$ large enough, and affine in $T$, outside this finite set. Thus, the positive recurrence follows from a multi-step version of the Foster-Lyapunov theorem \cite[Theorem C.27]{DaiHarrisonBook}.

\subsection{Transience}
Let us define the process $\underline N(\cdot)$, with $\underline N(0)=N(0)$, and defined recursively as
\[ \underline N(t+1) = \underline N(t) + A(t) - \sum\limits_{\ell=1}^{\overline{\overline{M}}(t)} Y_\ell(t). \]
This is equivalent to always attempting three-way matchings, even if the requests queue is not positive. Therefore, the process $\underline{N}(\cdot)$ can become negative. It is easy to check that $N(t) \geq \underline N(t)$ for all $t\geq 0$. We then have
\begin{align*}
    N(t) &\geq \underline N(t) \\
    & = N(0) + \sum\limits_{k=0}^{t-1} A_1(k) - \sum\limits_{k=0}^{t-1} \sum\limits_{\ell=1}^{\overline{\overline{M}}(k)} Y_\ell(k)
\end{align*}
Since $\lambda > p C_Y$, the right hand side diverges to $+\infty$ almost surely as $t\to+\infty$, and thus so does $N(t)$. In particular, this means that $N(\cdot)$ is transient. Finally, the throughput is shown to be $p C_Y$ in this case using the same argument as in the proof of Lemma \ref{lem:trivialStability}.

\section{Proof of Proposition \ref{prop:equivalence}}\label{app:equivalence}

We have
\begin{align*}
    Q_1(t+1) + Q_3(t+1) &= Q_1(t) + Q_3(t) - D_1(t) - D_3(t) + S_1(t) + S_3(t) - M_1(t) - M_2(t) \\
    Q_2(t+1) &= Q_2(t) - D_2(t) +S_2(t) - M_1(t) - M_2(t),
\end{align*}
where
\begin{align*}
    M_1(t) &= \min\big\{Q_1(t)-D_1(t)+S_1(t),\,\, Q_2(t)-D_2(t)+S_2(t) - X(t)M_2(t) \big\} \\
    M_2(t) &= \min\big\{Q_2(t)-D_2(t)+S_2(t) - [1-X(t)]M_1(t),\,\, Q_3(t)-D_3(t)+S_3(t) \big\}.
\end{align*}
If $X(t)=0$, we have
\begin{align*}
    M_1(t) &= \min\big\{Q_1(t)-D_1(t)+S_1(t),\,\, Q_2(t)-D_2(t)+S_2(t) \big\} \\
    M_2(t) &= \min\big\{Q_2(t)-D_2(t)+S_2(t) - M_1(t),\,\, Q_3(t)-D_3(t)+S_3(t) \big\},
\end{align*}
and thus
\begin{align*}
    M_1(t) + M_2(t) & = \min\big\{Q_2(t)-D_2(t)+S_2(t),\,\, Q_3(t)-D_3(t)+S_3(t) + M_1(t) \big\} \\
    &\qquad = \min\Big\{Q_2(t)-D_2(t)+S_2(t),\,\, Q_3(t)-D_3(t)+S_3(t) \\
    &\qquad\qquad\qquad\qquad\qquad + \min\big\{Q_1(t)-D_1(t)+S_1(t),\,\, Q_2(t)-D_2(t)+S_2(t) \big\} \Big\} \\
    &\qquad = \min\big\{ Q_2(t)-D_2(t)+S_2(t),\,\, Q_1(t)+Q_3(t)-D_1(t)-D_3(t)+S_1(t)+S_3(t) \big\}.
\end{align*}
The same holds if $X(t)=1$. Finally, since the indicator of the abandonments of qubits from the first and third queues are exchangeable, we can write
\[ D_1(t)+D_3(t) = \sum\limits_{\ell=1}^{Q_1(t)+Q_3(t)} Z^{(4)}_\ell(t) \]
where $\{Z^{(4)}_\ell(t): \ell,t\in\mathbb{Z}_+\}$ is a set of exchangeable Bernoulli random variables with probability $\gamma_1$ such that, for each $\ell\geq 1$, $\{Z^{(4)}_\ell(t): t\geq 0\}$ is i.i.d..

\section{Proof of Theorem \ref{thm:uniquenessAndStability}} \label{app:uniquenessAndStability}

Suppose that $n_1(0)>0$ and $n_2(0)>0$. Since $\lambda_1+\lambda_2<C_{1,2}$, cases (a) and (b) in the drift imply that
    \[ \frac{dn_1(t)}{dt} + \frac{dn_2(t)}{dt} = \lambda_1 + \lambda_2 - C_{1,2} <0. \]
    Thus, at least one of the coordinates hits zero at a time
    \[ T_1 \leq \frac{\|{\bf n}(0)\|_1}{C_{1,2} - \lambda_1 - \lambda_2}. \]
    Note that, if both coordinates hit zero at the same time, they stay there forever (case (d) in the drift). If only the $i$-th trajectory hits zero while the other one is positive, then we must have $\lambda_i<\underline{C}_i$ (otherwise it would not had hit zero while the other one was positive). Therefore, the $i$-th trajectory stays at zero forever.
    
    At time $T_1$, we have
    \[ \max\big\{ n_1(T_1),\, n_2(T_1) \big\} \leq \|{\bf n}(0)\|_\infty + \max\left\{\lambda_1 - \underline{C}_1,\, \lambda_2 - \underline{C}_2\right\} T_1. \]
    Moreover, when one of the coordinates is zero, we have
    \[ \frac{dn_1(t)}{dt} + \frac{dn_2(t)}{dt} \leq \max\big\{\lambda_1 - C_1(\lambda_2),\, \lambda_2 - C_2(\lambda_1) \big\}<0. \]
    Then, the second coordinate hits zero (and also stays there forever) after
    \[ T_2 \leq \frac{\|{\bf n}(0)\|_\infty + \max\left\{\lambda_1 - \underline{C}_1,\, \lambda_2 - \underline{C}_2\right\} T_1}{\min\big\{C_1(\lambda_2)-\lambda_1,\, C_2(\lambda_1)-\lambda_2\big\}}. \]
    units of time after $T_1$. We conclude that $\|{\bf n}(t)\|_1=0$ for all $t$ larger than
    \begin{align*}
        T_1 + T_2 &\leq \frac{\|{\bf n}(0)\|_1}{C_{1,2} - \lambda_1 - \lambda_2} + \frac{\|{\bf n}(0)\|_\infty + \max\left\{\lambda_1 - \underline{C}_1,\, \lambda_2 - \underline{C}_2\right\} T_1}{\min\big\{C_1(\lambda_2)-\lambda_1,\, C_2(\lambda_1)-\lambda_2\big\}} \\
        &\leq \|{\bf n}(0)\|_1 \left[ \frac{1}{C_{1,2} - \lambda_1 - \lambda_2} + \frac{1}{\min\big\{C_1(\lambda_2)-\lambda_1,\, C_2(\lambda_1)-\lambda_2\big\}}\left( 1 + \frac{\max\left\{\lambda_1 - \underline{C}_1,\, \lambda_2 - \underline{C}_2\right\}}{C_{1,2} - \lambda_1 - \lambda_2} \right) \right].
    \end{align*}

\section{Proof of Theorem \ref{thm:fluid}}\label{app:fluid}
Without loss of generality, we assume that $T\in\mathbb{Z}_+$. Moreover, we use a slight abuse of notation and assume that our stochastic processes are c\`adl\`ag functions over the real line, piece-wise constant over all intervals $[t,t+1)$.

\subsection{Tightness of sample paths}
First, using a simple comparison with appropriately coupled discrete-time GI/M/$\infty$ queues (cf. Equation \eqref{eq:collapseOfQ}), it follows that
\begin{equation}
\mathbb{P} \left(\lim\limits_{k\to\infty} \left\| {\bf Q}^{(k)}(\omega,t) \right\|_1 = 0, \,\, \forall\, t \in [0,T] \right) = 1. \label{eq:Qcollapse}
\end{equation}

We define the family of coupled processes $\left\{\left( {\bf Q}^{[x]}(\cdot),{\bf N}^{[x]}(\cdot),X^{[x]}(\cdot),{\bf M}^{[x]}(\cdot) \right) : x \in\{0,1,2,3\}\right\}$ such that they are constructed using the same stochastic primitives as the original processes ${\bf Q}(\cdot)$, ${\bf N}(\cdot)$, $X(\cdot)$, and ${\bf M}(\cdot)$, except the following modifications:
\begin{itemize}
    \item[-] {\bf Case $x=0$:} We have $N^{[0]}_1(t)=N^{[0]}_2(t)=\infty$ and $X^{[0]}(t)=0$ for all $t\geq 0$.
    \item[-] {\bf Case $x=1$:} We have $N^{[1]}_1(t)=N^{[1]}_2(t)=\infty$ and $X^{[1]}(t)=1$ for all $t\geq 0$.
    \item[-] {\bf Case $x=2$:} We have $N^{[2]}_1(t)=\infty$ for all $t\geq 0$.
    \item[-] {\bf Case $x=3$:} We have $N^{[3]}_2(t)=\infty$ for all $t\geq 0$.
\end{itemize}
Cases $2$ and $3$ are only considered when $\lambda_2<\underline{C}_2$ and $\lambda_1<\underline{C}_1$, respectively.

\begin{lemma}\label{lem:nice_set}
    Fix $T>0$. There exists a measurable set $\mathcal{C}$ such that $\mathbb{P}(\mathcal{C})=1$ and such that, for all $\omega\in\mathcal{C}$, we have
    \begin{align}
        &\lim\limits_{k\to\infty} \left\| Q^{(k)}(\omega,t) \right\|_1 = 0, \qquad \forall\, t \in [0,T], \\
        &\lim\limits_{k\to\infty} \sup\limits_{t\in[0,T]} \left| \frac{1}{k} \sum\limits_{s=1}^{\lfloor kt \rfloor} A_i(\omega,s) - \lambda_i t \right| = 0 \\
        &\lim\limits_{k\to\infty} \sup\limits_{t\in[0,T]} \left| \frac{1}{k} \sum\limits_{s=1}^{\lfloor kt \rfloor} S_j(\omega,s) - \mu_j t \right| = 0,
    \end{align}
    and, for any $0\leq t_1<t_2\leq T$, we have
    \begin{align}
        &\lim\limits_{k\to\infty} \frac{1}{k} \sum\limits_{j=\lfloor k t_1 \rfloor + 1}^{\lfloor k t_2 \rfloor} \sum\limits_{\ell=1}^{M^{[x]}_1(\omega,j)} Y_\ell^{(1)}(\omega,t) = (t_2-t_1)\Big(\overline{C}_1 \mathds{1}_{\{x=0\}} + C_1(\lambda_2) \mathds{1}_{\left\{x=2\right\}} + \underline{C}_1 \mathds{1}_{\{x=1\}} + \lambda_1 \mathds{1}_{\{x=3\}} \Big) \label{eq:convergence1} \\
        &\lim\limits_{k\to\infty} \frac{1}{k} \sum\limits_{j=\lfloor k t_1 \rfloor + 1}^{\lfloor k t_2 \rfloor} \sum\limits_{\ell=1}^{M^{[x]}_2(\omega,j)} Y_\ell^{(2)}(\omega,t) = (t_2-t_1)\Big(\overline{C}_2 \mathds{1}_{\{x=1\}} + C_2(\lambda_1) \mathds{1}_{\left\{x=3\right\}} + \underline{C}_2 \mathds{1}_{\{x=0\}} + \lambda_2 \mathds{1}_{\{x=2\}} \Big), \label{eq:convergence2}
    \end{align}
    for all $x\in\{0,1\}$, and for $x=2$ with $\lambda_2<\underline{C}_2$ and $x=3$ with $\lambda_1<\underline{C}_1$.
\end{lemma}
\begin{proof}{Proof:}
The first three equations hold in a set of probability one due to Equation \eqref{eq:Qcollapse}, and the Functional Strong Law of Large Numbers.\\

We have
\begin{align*}
    M^{[x]}_1(t) &= \min\left\{Q^{[x]}_1(t)-\sum\limits_{\ell=1}^{Q^{[x]}_1(j)} Z_\ell^{(1)}(j)+S_1(t),\,\, Q^{[x]}_2(t)-\sum\limits_{\ell=1}^{Q^{[x]}_2(j)} Z_\ell^{(2)}(j)+S_2(t) - X^{[x]}(t)M^{[x]}_2(t) \right\}, \\
    M^{[x]}_2(t) &= \min\left\{Q^{[x]}_2(t)-\sum\limits_{\ell=1}^{Q^{[x]}_2(j)} Z_\ell^{(2)}(j)+S_2(t) - [1-X^{[x]}(t)]M^{[x]}_1(t),\,\, Q^{[x]}_3(t)-\sum\limits_{\ell=1}^{Q^{[x]}_3(j)} Z_\ell^{(3)}(j)+S_3(t) \right\}.
\end{align*}
For $x\in\{1,2\}$, since each of the coordinates of ${\bf Q}^{[x]}(\cdot)$ is upper bounded by an ergodic GI/M/$\infty$ queue with correlated services (Lemma \ref{lem:infinityQueuesBound}), then the process $\big({\bf N}^{[x]}(\cdot),{\bf Q}^{[x]}(\cdot)\big)$ is ergodic. For $x=2$ with $\lambda_2<\underline{C}_2$ and $x=3$ with $\lambda_1<\underline{C}_1$, Lemma \ref{lem:halfStability} implies that the process $\big({\bf N}^{[x]}(\cdot),{\bf Q}^{[x]}(\cdot)\big)$ is ergodic. 

Fix $t_1,t_2\in\mathbb{Q}\cap[0,T]$. Since $\big({\bf N}^{[x]}(\cdot),{\bf Q}^{[x]}(\cdot)\big)$ is ergodic, for all sample paths in a set $\mathcal{C}_{\{t_1,t_2\}}$ such that $\mathbb{P}\big(\mathcal{C}_{\{t_1,t_2\}}\big)=1$, equations \eqref{eq:convergence1} and \eqref{eq:convergence2} hold for times $t_1$ and $t_2$. Thus, equations \eqref{eq:convergence1} and \eqref{eq:convergence2} hold for all $t_1,t_2\in\mathbb{Q}\cap[0,T]$ for all sample paths in the set
\[ \mathcal{C}_{\mathbb{Q}\cap[0,T]} = \bigcap\limits_{t_1,t_2\in \mathbb{Q}\cap[0,T]} \mathcal{C}_{\{t_1,t_2\}}, \]
which also has probability one because it is a countable intersection. Finally, it is easily checked that these equations \eqref{eq:convergence1} and \eqref{eq:convergence2} also hold for all $t_1,t_2\in[0,T]$ over the same set by the monotonicity of the sum and the continuity of the limit.   \end{proof}

Let us fix an arbitrary ${\bf x^0}\in\mathbb{R}_+^2$, sequences $B_k \downarrow 0$ and $\beta_k\downarrow 0$, and a constant $L>0$. For $k\geq 1$, we define the following subsets of $D^2[0,T]$. Let
\[  E_k(B_k,\beta_k) = \Big\{ {\bf x}\in D^2[0,T]:\,\left\|{\bf x}(0)-{\bf x^0}\right\|_1\leq B_k, \text{ and } \|{\bf x}(a)-{\bf x}(b)\|_1\leq L|a-b| +\beta_k, \,\,\forall \, a,b\in[0,T]\Big\} \]
be the set of all c\`adl\`ag, piece-wise constant functions that are $\beta_k$-approximate $L$-Lipschitz continuous, and that their initial conditions are at most $B_k$ away from ${\bf x^0}$. We also define
\[ E_c = \Big\{ {\bf x}\in D^2[0,T]:\, {\bf x}(0)={\bf x^0},\,\, \|{\bf x}(a)-{\bf x}(b)\|_1\leq L|a-b|, \,\,\forall \, a,b\in[0,T]\Big\}, \]
which is the set of $L$-Lipschitz continuous functions with initial condition ${\bf x^0}$, which is known to be sequentially compact by the Arzel\`a-Ascoli theorem.

\begin{lemma}\label{lem:sequences}
  Fix $T>0$,  $\omega\in\mathcal{C}$, and some ${\bf n^0}\in\mathbb{R}_+^2$. Suppose that
  \[ \left\|{\bf N}^{(k)}(\omega,0)-{\bf n^0}\right\|_1 \leq B_k, \]
  for some sequence $B_k \downarrow 0$. Then, there exists a sequence $\beta_k \downarrow 0$ such that
  \[ {\bf N}^{(k)}(\omega,\cdot) \in E_k\left(B_k,\beta_k\right), \quad \forall\ k\geq 1, \]
with the constant $L$ in the definition of $E_k$ equal to $\lambda_1+\lambda_2+\mu_1+2\mu_2+\mu_3$.
\end{lemma}
\begin{proof}{Proof:}
  By construction, and the triangular inequality, we have
  \begin{align*}
   \left|N^{(k)}_1(\omega,a)-N^{(k)}_1(\omega,b)\right| &\leq \frac{1}{k}  \left| \sum\limits_{s=\lfloor kb \rfloor + 1}^{\lfloor ka \rfloor} A_1(s) \right| + \frac{1}{k}\left| \sum\limits_{s=\lfloor kb \rfloor + 1}^{\lfloor ka \rfloor} S_1(s) + S_2(s) \right| \\
   &\qquad\qquad\qquad\qquad\qquad\qquad\qquad + \frac{1}{k}\Big| Q_1(ka) - Q_1(kb) + Q_2(ka) - Q_2(kb) \Big|
    \end{align*}
  Since $\omega\in\mathcal{C}$, Lemma \ref{lem:nice_set} implies that there exists sequences $\beta_k^1\downarrow 0$, $\beta_k^2 \downarrow 0$, and $\beta_k^3 \downarrow 0$ (which depend on $\omega$) such that for all $k\geq 1$,
  \[ \frac{1}{k} \left| \sum\limits_{s=\lfloor kb \rfloor + 1}^{\lfloor ka \rfloor} A_1(s) \right| \leq \lambda_1|a-b|+\beta_k^1, \]
  \[ \frac{1}{k}\left| \sum\limits_{s=\lfloor kb \rfloor + 1}^{\lfloor ka \rfloor} S_1(s) + S_2(s) \right| \leq (\mu_1+\mu_2)|a-b|+\beta_k^2, \]
  and
  \[ \frac{1}{k}\Big| Q_1(ka) - Q_1(kb) + Q_2(ka) - Q_2(kb) \Big| \leq \beta_k^3, \]
  which imply
  \[ \left|N^{(k)}_1(\omega,a)-N^{(k)}_1(\omega,b)\right| \leq (\lambda_1+\mu_1+\mu_2)|a-b|+\left(\beta_k^1+\beta_k^2+\beta_k^3\right).
 \]
 Analogously we can show that
 \[ \left|N^{(k)}_2(\omega,a)-N^{(k)}_2(\omega,b)\right| \leq (\lambda_2+\mu_2+\mu_3)|a-b|+\left(\beta_k^4+\beta_k^5+\beta_k^6\right).
 \]
  The proof is completed by setting $\beta_k=\beta_k^1+\beta_k^2+\beta_k^3+\beta_k^4+\beta_k^5+\beta_k^6$, and $L=\lambda_1+\lambda_2+\mu_1+2\mu_2+\mu_3$.
  \end{proof}

We are now ready to prove the existence of convergent subsequences of the process of interest.

\begin{proposition}\label{prop:tightness}
  Fix $T>0$, $\omega\in\mathcal{C}$, and some ${\bf n^0}\in\mathbb{R}_+^2$. Suppose (as in Lemma \ref{lem:sequences}) that
  \[ \left\|{\bf N}^{(k)}(\omega,0)-{\bf n^0}\right\|_1\leq B_k, \]
  where $B_k\downarrow 0$. Then, every subsequence of $\left\{{\bf N}^{(k)}(\omega,\cdot)\right\}_{k=1}^{\infty}$ contains a further subsequence $\left\{{\bf N}^{(k_\ell)}(\omega,\cdot)\right\}_{\ell=1}^{\infty}$ that converges to a coordinate-wise Lipschitz continuous function ${\bf n}(\cdot)$ with ${\bf n}(0)={\bf n^0}$ and
  \[ |n_i(a)-n_i(b)| \leq L|a-b|, \quad\quad \forall \, a,b\in[0,T], \]
  where $L$ does not depend on $T$, $\omega$, nor on ${\bf n}(\cdot)$.
\end{proposition}
\begin{proof}{Proof:}
As in Lemma \ref{lem:sequences}, let $L=\lambda_1+\lambda_2+\mu_1+2\mu_2+\mu_3$. Note that, if the trajectories were $L$-Lipschitz continuous, then the Arzel\`a-Ascoli would immediately give us this result. However, since trajectories are only approximately $L$-Lipschitz continuous (Lemma \ref{lem:sequences}), a more involved argument is required. This argument was introduced in \cite{bramson98} and further developed in \cite{powerOfLittle} (and it is therefore omitted), and involves doing a continuous linear interpolation of the piece-wise constant to prove the ``closeness'' of $E_k\big(B_k,\beta_k\big)$ to the sequentially compact set $E_c$.
  \end{proof}

\subsection{Derivatives of the limit points}

In this subsection we will show that the limit points of the sequence of sample paths (which we showed that they are Lipschitz continuous in Proposition \ref{prop:tightness}) have the same derivatives as a fluid solution (cf. Definition \ref{def:fluid}) for all regular points.

\begin{proposition}\label{prop:derivatives}
  Fix $\omega\in\mathcal{C}$ and $T>0$. Let ${\bf n}(\cdot)$ be a limit point of some subsequence of $\left\{{\bf N}^{(k)}(\omega,\cdot)\right\}_{k=1}^\infty$. Then, ${\bf n}(\cdot)$ has the same derivatives as a fluid solution (cf. Definition \ref{def:fluid}) for all regular points.
\end{proposition}
\begin{proof}{Proof:}
We fix some $\omega\in\mathcal{C}$ and for the rest of this proof we suppress the dependence on $\omega$ in our notation. Let $\left\{{\bf N}^{(k_\ell)}(\cdot)\right\}_{\ell=1}^\infty$ be a subsequence such that
\[ \lim\limits_{\ell\to\infty} \sup\limits_{0\leq t\leq T}\left\|{\bf N}^{(k_\ell)}(t)-{\bf n}(t)\right\|_1=0, \qquad a.s. \]
After possibly restricting, if necessary, to a further subsequence, we can define Lipschitz continuous functions $a_i(\cdot)$ and $r_i(\cdot)$ as the limits of the subsequences of cumulative arrivals and cumulative matching processes $\big\{A_i^{(k_\ell)}(\cdot)\big\}_{\ell=1}^\infty$ and $\big\{R_i^{(k_\ell)}(\cdot)\big\}_{\ell=1}^\infty$ respectively, where
\[ R_i^{(k_\ell)}(t) = \frac{1}{k_\ell} \sum\limits_{j=1}^{\lfloor k_\ell t \rfloor} \sum\limits_{\ell=1}^{M_1(j)} Y_\ell^{(1)}(j). \]
Because of the relation $N_i^{(k)}(t)=N_i^{(k)}(0)+A_i^{(k)}(t)-R_i^{(k)}(t)$, it is enough to prove the following relations, for almost all $t$:
\begin{align*}
    \frac{da_i(t)}{dt}=&\lambda_i \qquad \text{and} \qquad \frac{dr_i(t)}{dt}= R_i(t).
  \end{align*}
We will provide a proof only for the second one, as the first one is immediate by the Functional Strong Law of Large Numbers.\\

Let us fix some time $t\in[0,T)$, which is a regular time for both $r_1(\cdot)$ and $r_2(\cdot)$. Let $\epsilon>0$ be small enough so that $t+\epsilon \leq T$ and so that it also satisfies a condition to be introduced later. By construction, we have
\begin{align*}
  R^{(k_\ell)}_1(t+\epsilon)-R^{(k_\ell)}_1(t) &= \frac{1}{k_\ell} \sum\limits_{j=\lfloor k_\ell t \rfloor + 1}^{\lfloor k_\ell(t+\epsilon)\rfloor} \sum\limits_{\ell=1}^{M_1(j)} Y_\ell^{(1)}(j), \\
  R^{(k_\ell)}_2(t+\epsilon)-R^{(k_\ell)}_2(t) &= \frac{1}{k_\ell} \sum\limits_{j=\lfloor k_\ell t \rfloor + 1}^{\lfloor k_\ell(t+\epsilon)\rfloor} \sum\limits_{\ell=1}^{M_2(j)} Y_\ell^{(2)}(j),
\end{align*}
where
\begin{align}
  M_1(j) &= \min\Bigg\{Q_2(j)-D_2(j)+S_2(j) - X(j)M_2(j),\,\, Q_1(j)-D_1(j)+S_1(j), \nonumber \\
  &\qquad\qquad\qquad\quad\qquad\qquad\qquad\qquad\qquad\quad \left. \min\left\{ m_1 : \sum\limits_{\ell=1}^{m_1} Y_\ell^{(1)}(j) \geq N_1(j)+A_1(j) \right\} \right\} \label{eq:R1} \\
  M_2(j) &= \min\Bigg\{Q_3(j)-D_3(j)+S_3(j) - [1-X(j)]M_1(j),\,\,  Q_2(j)-D_2(j)+S_2(j), \nonumber \\
  &\qquad\qquad\qquad\quad\qquad\qquad\qquad\qquad\qquad\quad \left. \min\left\{ m_2 : \sum\limits_{\ell=1}^{m_2} Y_\ell^{(2)}(j) \geq N_2(j)+A_2(j) \right\} \right\}. \label{eq:R2}
\end{align}
By Lemma \ref{lem:sequences}, there exists a sequence $\beta_{k_\ell} \downarrow 0$ and a constant $L$ such that
 \begin{equation*}
 N_i^{(k_\ell)}(u) \in \big[n_i(t)-\left(\epsilon L+\beta_{k_\ell}\right),\ n_i(t)+\left(\epsilon L+\beta_{k_\ell}\right)\big), \quad \forall\,u\in[t,t+\epsilon]
\end{equation*}
Then, for all sufficiently large $\ell$, we have
 \begin{equation}
 N_i^{(k_\ell)}(u) \in \big[n_i(t)-2\epsilon L,\ n_i(t)+2\epsilon L)\big), \quad  \forall\,u\in[t,t+\epsilon] \label{eq:set_inclusion}
\end{equation}
and
\begin{equation*}
 N_i^{(k_\ell)}(u) + \frac{1}{k_\ell} A_i(\lfloor u \rfloor ) \in \big[n_i(t)-2\epsilon L,\ n_i(t)+2\epsilon L)\big), \quad  \forall\,u\in[t,t+\epsilon].
\end{equation*}

\noindent {\bf Case (a): $n_1(t)>n_2(t)>0$}\\

For $\epsilon$ small enough, we have $n_1(t)-2\epsilon L > n_2(t) + 2\epsilon L > 0$, and thus $N^{(k_\ell)}_1(u)>N^{(k_\ell)}_2(u)>0$ for all $u\in[t,t+\epsilon]$, for all $\ell$ large enough. It follows that $X(j)=0$ for all $j=\lfloor k_\ell t \rfloor + 1,\dots,\lfloor k_\ell(t+\epsilon)\rfloor$, and that the minima in equations \eqref{eq:R1} and \eqref{eq:R2} are never attained in the third case (otherwise we would have $N^{(k_\ell)}_2(u)=0$ for some $u\in[t,t+\epsilon]$). Therefore, we have
\begin{align*}
  M_1(j) &= \min\Big\{Q_2(j)-D_2(j)+S_2(j),\,\, Q_1(j)-D_1(j)+S_1(j) \Big\}, \\
  M_2(j) &= \min\Big\{Q_3(j)-D_3(j)+S_3(j) - M_1(j),\,\,  Q_2(j)-D_2(j)+S_2(j) \Big\},
\end{align*}
where $Q(\cdot)$ is defined recursively as
\begin{align*}
    Q_1(u+1) &= Q_1(u) - D_1(u) + S_1(u) - M_1(u) \\
    Q_2(u+1) &= Q_2(u) - D_2(u) +S_2(u) - M_1(u) - M_2(u) \\
    Q_3(u+1) &= Q_3(u) - D_3(u) +S_3(u) - M_2(u).
\end{align*}
Note that this corresponds to Case 1 in Lemma \ref{lem:nice_set}. Therefore, since $\omega\in\mathcal{C}$, Lemma \ref{lem:nice_set} implies
\begin{align*}
    r_1(t+\epsilon) - r_1(t) &= \lim\limits_{\ell\to\infty} R^{(k_\ell)}_1(t+\epsilon)-R^{(k_\ell)}_1(t) \\
    &= \epsilon .\, \overline{C}_1,
\end{align*}
and
\begin{align*}
    r_2(t+\epsilon) - r_2(t) &= \lim\limits_{\ell\to\infty} R^{(k_\ell)}_2(t+\epsilon)-R^{(k_\ell)}_2(t) \\
    &= \epsilon .\, \underline{C}_2.
\end{align*}
Dividing by $\epsilon$ and taking the limit as $\epsilon$ goes to zero we conclude that, when $n_1(t)>n_2(t)>0$, we have
\begin{align*}
    \frac{dn_1(t)}{dt} &= \lambda_1 - \overline{C}_1 \qquad \text{and} \qquad \frac{dn_2(t)}{dt} = \lambda_2 - \underline{C}_2.
\end{align*}

\noindent {\bf Case (b): $n_1(t)=n_2(t)>0$}\\

First note that, if $n_1(t)>0$ and $n_2(t)>0$, the same argument as in the previous case gives us
\begin{equation}
    r_i(t+\epsilon)-r_i(t) \in \big[\epsilon\underline{C}_i,\, \epsilon\overline{C}_i\big], \label{eq:rateBounds}
\end{equation}
and
\begin{equation}
r_1(t+\epsilon)+r_2(t+\epsilon)-r_1(t)-r_2(t) = \epsilon C_{1,2}, \label{eq:rateTotal}
\end{equation}
for all sufficiently small $\epsilon$.\\

Suppose that
\[ \lambda_1 - \overline{C}_1 > \lambda_2 - \underline{C}_2. \]

Combining this with Equation \eqref{eq:rateBounds}, it follows that 
\[ n_1(u) > n_2(u) \]
for all $u\in(t,t+\epsilon]$. Therefore
\begin{align*}
    \frac{dn_1(u)}{du} &= \lambda_1 - \overline{C}_1 \qquad \text{and} \qquad \frac{dn_2(u)}{du} = \lambda_2 - \underline{C}_2,
\end{align*}
for all $u\in(t,t+\epsilon]$. Since $t$ is a regular time, we also have
\begin{align*}
    \frac{dn_1(t)}{dt} &= \lambda_1 - \overline{C}_1 \qquad \text{and} \qquad \frac{dn_2(t)}{dt} = \lambda_2 - \underline{C}_2.
\end{align*}

Analogously, if 
\[ \lambda_2 - \overline{C}_2 > \lambda_1 - \underline{C}_1, \]
we have
\begin{align*}
    \frac{dn_1(t)}{dt} &= \lambda_1 - \underline{C}_1 \qquad \text{and} \qquad \frac{dn_2(t)}{dt} = \lambda_2 - \overline{C}_2.
\end{align*}

Finally, suppose that
\[ \frac{\lambda_1 + \lambda_2 - C_{1,2}}{2} \geq \max\left\{ \lambda_1 - \overline{C}_1,\,\, \lambda_2 - \overline{C}_2 \right\}. \]
In particular, this means that 
\[ \lambda_1 - \overline{C}_1 \leq \lambda_2 - \underline{C}_2 \qquad \text{and} \qquad \lambda_2 - \overline{C}_2 \leq \lambda_1 - \underline{C}_1. \]
Therefore, we have 
\[ \frac{dn_1(u)}{du} - \frac{dn_2(u)}{du} \leq 0 \]
when $n_1(u)>n_2(u)$, and 
\[ \frac{dn_1(u)}{du} - \frac{dn_2(u)}{du} \geq 0 \]
when $n_1(u)<n_2(u)$.

\begin{lemma}\label{lem:absurd}
We have
\[ \frac{dn_1(t)}{dt} - \frac{dn_2(t)}{dt} =0, \]
and
\[ \frac{dn_1(t)}{dt} = \frac{dn_2(t)}{dt} = \frac{\lambda_1 + \lambda_2 - C_{1,2}}{2}. \]
\end{lemma}
\begin{proof}{Proof:}
We prove this by contradiction. Suppose that
\[ \frac{dn_1(t)}{dt} - \frac{dn_2(t)}{dt} > 0. \]
Since $n_1(t)=n_2(t)$ and $t$ is a regular point, we must have $n_1(u) > n_2(u)$ for all sufficiently small $u>t$. However, we have
\[ \frac{dn_1(u)}{du} - \frac{dn_2(u)}{du} \leq 0 \]
for all sufficiently small $u>t$, which contradicts the regularity of $t$. Assuming
\[ \frac{dn_1(t)}{dt} - \frac{dn_2(t)}{dt} < 0 \]
yields the same contradiction. Therefore, we must have
\[ \frac{dn_1(t)}{dt} - \frac{dn_2(t)}{dt} =0. \]
Combining this with Equation \eqref{eq:rateTotal}, we obtain
\[ \frac{dn_1(t)}{dt} = \frac{dn_2(t)}{dt} = \frac{\lambda_1 + \lambda_2 - C_{1,2}}{2}. \]
  \end{proof}

\noindent {\bf Case (c): $n_1(t)>n_2(t)=0$}\\

Suppose that $\lambda_2 < \underline{C}_2$. Then, using the same argument as in Case (a) but now with only $N_1(\cdot)$ being infinitely backlogged and $N_2(\cdot)$ being stable (Lemma \ref{lem:halfStability}), we conclude that
\begin{align*}
    \frac{dn_1(t)}{dt} &= \lambda_1 - C_1(\lambda_2) \qquad \text{and} \qquad \frac{dn_2(t)}{dt} = 0.
\end{align*}

Suppose that $\lambda_2 > \underline{C}_2$. Since $n_2(u)<n_1(u)$ for all sufficiently small $u>t$, we have
\begin{align}
    \frac{dr_2(t)}{dt} &\leq \underline{C}_2. \label{eq:lastDriftBound}
\end{align}
Therefore, we have $n_2(u)>0$ for all sufficiently small $u>t$. By Case (a), we have
\begin{align*}
    \frac{dn_1(u)}{du} &= \lambda_1 - \overline{C}_1 \qquad \text{and} \qquad \frac{dn_2(u)}{du} = \lambda_2 - \underline{C}_2
\end{align*}
for all sufficiently small $u>t$. Using once more that $t$ is a regular time, it follows that
\begin{align*}
    \frac{dn_1(t)}{dt} &= \lambda_1 - \overline{C}_1 \qquad \text{and} \qquad \frac{dn_2(t)}{dt} = \lambda_2 - \underline{C}_2.
\end{align*}

Finally, suppose that $\lambda_2 = \underline{C}_2$. Using the same argument as in Lemma \ref{lem:absurd}, it can be checked that
\begin{align*}
    \frac{dn_2(t)}{dt} &= 0.
\end{align*}

\noindent {\bf Case (d): $n_1(t)=n_2(t)=0$}\\

Since $t$ is a regular point, $n_i(0)>0$, and cases (a), (b), and (c) imply that $n_i(\cdot)$ can only become $0$ at a non-regular point, we must have $n_1(u)=n_2(u)=0$ for all sufficiently large $u<t$. It follows that
\[ \frac{dn_1(u)}{du} = \frac{dn_2(u)}{du} = 0 \]
for all sufficiently large $u<t$. Using once again that $t$ is a regular point yields
\[ \frac{dn_1(t)}{dt} = \frac{dn_2(t)}{dt} = 0. \]
  \end{proof}

\subsection{Completing the proof of Theorem \ref{thm:fluid}}

For every sample path in $\mathcal{C}$, we have established the following. Proposition~\ref{prop:tightness} implies the existence of limit points of the sequence of processes $\left\{N^{(k)}(\cdot)\right\}_{k=1}^\infty$. Furthermore, according to Proposition~\ref{prop:derivatives} these limit points verify the differential equations of the fluid model. Combining this with the fact that the trajectories are Lipschitz continuous, and thus they are differentiable almost everywhere, the limit points are fluid solutions. In particular, this means that fluid solutions exist.

Finally, since the trajectories are piece-wise linear and any component that hits zero stays at zero forever, the uniqueness of solutions follows from the uniqueness of the pieces that comprise them. In particular, this implies that all limit points are the same, and therefore the limit converges.

\section{Proof of Theorem \ref{thm:stability}}\label{app:stability}

\begin{itemize}
    \item[(i)] When $\lambda_1+\lambda_2<C_{1,2}$, $\lambda_1<C_1(\lambda_2)$, and $\lambda_2<C_2(\lambda_1)$, theorems \ref{thm:fluid} and \ref{thm:uniquenessAndStability} imply that there exists $\delta>0$ such that, for any $n^0\in\mathbb{R}_+^2$ with $\|n^0\|_1>0$, if $N(0)=kn^0$, we have
    \[ \lim\limits_{k\to\infty} \frac{N_i(kt)}{k} = 0, \qquad a.s., \]
    for all $t\geq \delta \|n^0\|_1$. Moreover, Equation \eqref{eq:collapseOfQ} implies that, for any $q^0\in\mathbb{R}_+^3$ with $\|q^0\|_1>0$, if $Q(0)=k q^0$, we have
    \[ \lim\limits_{k\to\infty} \frac{Q_j(kt)}{k} = 0, \qquad a.s., \]
    for all $t>0$. Therefore, the positive recurrence of $\big({\bf N}(\cdot),{\bf Q}(\cdot)\big)$ when $\lambda_1+\lambda_2<C_{1,2}$, $\lambda_1<C_1(\lambda_2)$, and $\lambda_2<C_2(\lambda_1)$ follows from these finite-time convergences to $0$, and \cite[Theorem 6.2]{DaiHarrisonBook}.
    \item[(ii)] We first consider a coupled process $\left(\underline N(\cdot), \underline{\bf  Q}(\cdot)\right)$, where the coupling with the original processes is done in a way so that $\underline N(0) = {\bf N}_1(0)+ {\bf N}_2(0)$, $\underline{\bf Q}_1(0)={\bf Q}_1(0)+{\bf Q}_3(0)$, and $\underline{\bf Q}_2(0)={\bf Q}_2(0)$ and so that they have the same arrival processes $A_1(\cdot)$, $A_2(\cdot)$, $S_1(\cdot)$, $S_2(\cdot)$, and $S_3(\cdot)$, and the same abandonment primitives $Z_\ell^{(i)}(\cdot)$. These new processes are defined recursively as
    \begin{align*}
        \underline N(t+1) &= \underline N(t) + A_1(t) + A_2(t) - \sum\limits_{\ell=1}^{\overline{\overline{M}}(t)} Y_\ell(t) \\
        \underline Q_1(t+1) &= \underline Q_1(t) - \underline D_1(t) + S_1(t) + S_3(t) - \overline{\overline{M}}(t) \\
        \underline Q_2(t+1) &= \underline Q_2(t) - \underline D_2(t) +S_2(t) - \overline{\overline{M}}(t),
    \end{align*}
    where
    \[ \overline{\overline{M}}(t) = \min\left\{\underline Q_1(t)-\underline D_1(t)+S_1(t) + S_3(t),\,\, \underline Q_2(t)-\underline D_2(t)+S_2(t)\right\}. \]
    This corresponds to a system where three-way matchings are always attempted if there are enough qubits, regardless of the requests available. Then, $\underline N(\cdot)$ keeps track of the difference between the total number of requests that arrived to the system, and the total number of successful three-way matchings. This can be negative, and it can be checked that $\underline N(t) \leq {\bf N}_1(t)+ {\bf N}_2(t)$, for all $t\geq 0$, for any non-idling policy. Moreover, since $\underline Q_1(\cdot)$ and $\underline Q_2(\cdot)$ behave as two-sided queues with abandonments, we have
    \[ \lim\limits_{t\to\infty} \frac{1}{t} \sum\limits_{s=0}^t \sum\limits_{\ell=1}^{\overline{\overline{M}}(s)} Y_\ell(s) = C_{1,2}. \]
    Therefore, since $\lambda_1+\lambda_2 > C_{1,2}$, we have that $\underline N(t)$ diverges to $+\infty$ almost surely and, since $\underline N(t) \leq {\bf N}_1(t)+ {\bf N}_2(t)$ for all $t\geq 0$, the process $\big({\bf N}(\cdot),{\bf Q}(\cdot)\big)$ is transient.\\

    On the other hand, suppose that $\lambda_1+\lambda_2 \leq C_{1,2}$, and that $\lambda_j < \underline{C}_j$ and $\lambda_i>C_i(\lambda_j)$. Analogously to the previous case, we can construct a coupled process $\underline N_i(\cdot)$ such that $N_i(t) \geq \underline N_i(t)$ for all $t\geq 0$, and such that its throughput is $C_i(\lambda_j)$. Then, since $\lambda_i>C_i(\lambda_j)$, we have that $\underline N_i(\cdot)$ diverges to $+\infty$ almost surely and, since $N_i(t) \geq \underline N_i(t)$ for all $t\geq 0$, the process $\big({\bf N}(\cdot),{\bf Q}(\cdot)\big)$ is transient.
    \end{itemize}
    
    \end{appendix}

\end{document}